\documentclass[a4paper,10pt,twoside]{article}

\usepackage[utf8]{inputenc} 
\usepackage[english]{babel} 
\usepackage[english=american]{csquotes}
\MakeOuterQuote{"}

\usepackage{authblk}

\usepackage{lmodern}
\usepackage{amssymb,amsthm,mathtools,enumitem} 
\usepackage{hyperref} 
\usepackage{geometry} 

\usepackage{graphicx}

\newcommand{\im}{\mathrm{i}}

\newcommand{\N}{\mathbb{N}}
\newcommand{\R}{\mathbb{R}}

\newcommand{\C}{\mathbb{C}}
\newcommand{\defeq}{\coloneqq}

\DeclareMathOperator{\id}{id}
\newcommand{\xd}{\mathrm{d}}

\newcommand{\cH}{\mathcal{H}}
\newcommand{\cD}{\mathcal{D}}
\DeclareMathOperator{\tr}{\mathrm{tr}}
\newcommand{\coh}{\mathsf{K}}
\newcommand{\one}{\mathbf{1}}
\newcommand{\Lop}{\mathrm{op}}
\newcommand{\xL}{\hat{L}}
\newcommand{\wq}[1]{\widehat{#1}}
\newcommand{\cM}{\mathcal{M}}
\newcommand{\cB}{\mathcal{B}}
\newcommand{\cT}{\mathcal{T}}
\newcommand{\nop}{\mathrm{op}}
\newcommand{\ntr}{\mathrm{tr}}
\newcommand{\comp}{\diamond}

\theoremstyle{definition}
\newtheorem{dfn}{Definition}[section]

\theoremstyle{plain}
\newtheorem{lem}[dfn]{Lemma}
\newtheorem{prop}[dfn]{Proposition}
\newtheorem{thm}[dfn]{Theorem}

\allowdisplaybreaks

\begin{document}


\begin{titlepage}
\title{\textbf{Spectral decomposition of field operators and causal measurement in quantum field theory}}
\author{Robert Oeckl\footnote{email: robert@matmor.unam.mx}}

\affil{Institute for Quantum Optics and Quantum Information, \\
Boltzmanngasse 3, 1090 Vienna, Austria}

\affil{Centro de Ciencias Matemáticas, \\
Universidad Nacional Autónoma de México, \\
C.P.~58190, Morelia, Michoacán, Mexico}

\date{UNAM-CCM-2024-2\\ 13 September 2024\\ 2 April 2025 (v2)}

\maketitle

\vspace{\stretch{1}}

\begin{abstract}

We construct the spectral decomposition of field operators in bosonic quantum field theory as a limit of a strongly continuous family of positive-operator-valued measure decompositions. The latter arise from integrals over families of bounded positive operators. Crucially, these operators have the same locality properties as the underlying field operators. We use the decompositions to construct families of quantum operations implementing measurements of the field observables. Again, the quantum operations have the same locality properties as the field operators. What is more, we show that these quantum operations do not lead to superluminal signaling and are possible measurements on quantum fields in the sense of Sorkin.

\end{abstract}

\vspace{\stretch{1}}
\end{titlepage}


\section{Introduction}

In non-relativistic quantum mechanics, self-adjoint operators on Hilbert space play an important role in encoding physical observables that are used to describe measurement processes and their outcomes. The eigenvalues of the observable are the possible values labeling the outcomes of a measurement, while the eigenspaces determine the state after the measurement. That is, the spectral decomposition of observables plays an essential role in quantum measurement theory. In the special case of performing only a single measurement to obtain an expectation value of an observable and discarding the system afterwards one can get away without knowing the spectral decomposition of the observable explicitly. However, in the case of multiple measurements or a single measurement with post-selection this is no longer the case. Indeed, in modern quantum measurement theory, the basic concept is not that of an observable, but that of a quantum operation \cite{Kra:statechanges}. This is a completely positive super-operator on the space of self-adjoint operators on Hilbert space. To construct the quantum operations corresponding to an observable requires precisely the spectral decomposition of the observable.

Historically, observables in terms of self-adjoint operators have played a much less important role in describing measurement processes in quantum field theory (QFT). Instead, scattering theory has been of most interest, described through an asymptotic transition amplitude. This corresponds to a one-shot measurement where preparation and observation happen at asymptotically early and late times respectively. While this has been adequate for most applications in high energy physics, there is a growing need for a full-fledged measurement theory of quantum field theory that is able to describe processes involving multiple measurements that are moreover localized in both time and space. Rather than provide individual references on this subject we recommend to the reader an excellent recent survey article by Papageorgiou and Fraser \cite{PaFr:eliminatingimpossible}.

The most straightforward approach to a measurement theory of quantum field theory would be to follow the same path as in non-relativistic quantum mechanics. That is, use the spectral decomposition of the observables of interest and construct the corresponding quantum operations. The literature is surprisingly sparse on this topic. What is more, there is an important body of literature suggesting that this approach is problematic. We limit ourselves here to mention Sorkin's seminal work on the subject, published in 1993 \cite{Sor:impossible}. In it, Sorkin showed that a relatively generic class of projective measurements leads to superluminal signaling and is thus unphysical. A quantum operation constructed from projection operators of a spectral decomposition would presumably be affected by this.

This difficulty has boosted efforts to look for alternative ways of describing measurement in QFT, particularly through an explicit modelling of the measurement apparatus as a quantum system in its own right which is subsequently measured in the conventional way, in the sense of von~Neumann \cite{vne:mathgrundquant}. One would then aim to show that the interaction between system and apparatus satisfies locality and does not lead to superluminal signaling. Two approaches have reported success in this respect recently. One is based on particle-detector models \cite{PGGaMM:detectormeasurementqft}, in the tradition of Unruh-DeWitt detectors \cite{BiDa:qftcurved}. The other one is based on modeling the apparatus through a quantum field \cite{FeVe:qftlocalmeasure}, generalizing the work of Hellwig and Kraus \cite{HeKr:opmeasureii}. (In the latter case discussion of Sorkin's problem is in \cite{BoFeRu:impossible,MuVe:superluminal}.) Note that these approaches still need to recur to the conventional non-relativistic measurement scheme to extract the measured value from the apparatus after the system-apparatus interaction, potentially reintroducing vulnerabilities.

The simplest, but arguably most important class of localizable observables in QFT is given by field operators. By these we mean here those operators arising from the quantization of linear functions on phase space. Equivalently, these are the operators that can be written as the sum of a creation and corresponding annihilation operator. Using coherent states, properties of the Weyl quantization, and characteristic functions of measures, we motivate in Section~\ref{sec:motivation} the sesquilinear form determining the spectral measure of field operators. In Section~\ref{sec:povmdec} we extend this sesquilinear form to a 1-parameter family and show that this gives rise to a 1-parameter family of positive-operator-valued measure (POVM) decompositions of field operators, satisfying continuity and composition properties, with the spectral measure as a limit. In Section~\ref{sec:quantop} we construct quantum operations to measure the field operator. These come in two types. The discrete outcome type, based on the spectral decomposition, consists of partitioning the real line into subsets, with quantum operations testing whether the value of the observable lies in a given subset or not (Section~\ref{sec:doutcomes}). The continuous outcome type amounts to a full and precise measurement of the value of the observable (Section~\ref{sec:coutcomes}) and is based instead on the 1-parameter family of POVM decompositions, approximating the spectral decomposition. Locality and causality properties of the quantum operations are investigated in Section~\ref{sec:loccaus}. In particular, we show that all quantum operations have the same locality properties as the corresponding field operator. Moreover, we show that measurements of the continuous outcome type do not lead to superluminal signaling, but satisfy what we call causal transparency. We conclude the paper with some discussion and an outlook in Section~\ref{sec:outlook}. The appendix contains parts of the proofs of Lemmas~\ref{lem:prod} and \ref{lem:projsumid}.


\section{Motivation}
\label{sec:motivation}


Consider a separable complex Hilbert space $L$ with inner product $\{\cdot,\cdot\}$. Denote by $L^*$ the dual Hilbert space. Let $\xL$ be the algebraic dual of $L^*$, equipped with the weak *-topology. Then, the natural inclusion $L\subseteq \xL$ is continuous. The pair $(L,\xL)$ forms an abstract Wiener space \cite{Gro:wienerspaces}.\footnote{The construction is slightly different to that given by Gross. In particular, $\hat{L}$ is merely a Fréchet space rather than a Banach space, see \cite{Oe:holomorphic}.} In particular, $\xL$ acquires a centered Gaussian probability measure $\nu$ determined by the inner product of $L$. Let $\cH$ be the separable Hilbert space of square-integrable holomorphic functions on $\xL$ with inner product $\langle\cdot,\cdot\rangle$ \cite{Oe:holomorphic}. Note that the elements of $\cH$ are completely determined by their values on $L$ \cite[Theorem~3.18]{Oe:holomorphic}. This allows for a treatment strongly analogous to the case where $L$ is finite-dimensional, even if it is not. $\cH$ is the space of holomorphic wave functions of the bosonic quantum field theory determined by the linear phase space $L$. The space $\cH$ is naturally isomorphic to the usual Fock space construction over $L$ \cite{Oe:freefermi}.

The space $\cH$ is a reproducing-kernel Hilbert space in the sense of Bargmann \cite{Bar:hilbanalytic} (generalized, if $L$ is infinite-dimensional). In particular, we have a family of coherent states $\{\coh_{\xi}\}_{\xi\in L}$ indexed by elements of $L$. Their wave functions are given by,
\begin{equation}
    \coh_\xi(\phi)=\exp\left(\frac12\{\xi,\phi\}\right), \quad\forall\xi\in L\, \forall\phi\in L .
\end{equation}
The reproducing property is,
\begin{equation}
    \psi(\phi)=\langle\coh_\phi,\psi\rangle,\quad\forall\phi\in L\,\forall \psi\in\cH.
\end{equation}
From this follows the completeness relation,
\begin{equation}
    \langle\eta,\psi\rangle=\int_{\xL}\xd\nu(\phi)\, \langle\eta,\coh_\phi\rangle \langle\coh_\phi,\psi\rangle,\quad\forall\psi,\eta\in\cH .
    \label{eq:cohcompl}
\end{equation}

In the following, we write the complex structure of $L$ always as a map $J:L\to L$. We may view $L$ as a real vector space and consider its complexification that we will denote $L^\C=L\oplus\im L$. Then, $J$ extends to a complex linear anti-involution $J:L^{\C}\to L^{\C}$. The inner product extends to a complex bilinear map on $L^{\C}$ and we decompose it into real and imaginary parts as follows \cite{Woo:geomquant},
\begin{equation}
    \{\xi,\phi\}=g(\xi,\phi)+2\im\omega(\xi,\phi)
    =2\omega(\xi,J \phi)+2\im\omega(\xi,\phi) .
    \label{eq:lip}
\end{equation}
Here, $g$ and $\omega$ are symmetric and antisymmetric non-degenerate real bilinear forms on $L$ respectively. Both extend to complex bilinear forms on $L^\C$. We refer to $\omega$ also as the symplectic form.

We proceed to consider the quantization of a linear observable on the phase space $L$. That is, consider a real linear and continuous map $L\to\R$. By the Riesz representation theorem, there exists a unique element $\xi\in L$ such that the linear map is given by
\begin{equation}
    D_{\xi}(\phi)=2\omega(\xi,\phi)=g(J\xi,\phi), \quad\forall\phi\in L .
    \label{eq:linobsdual}
\end{equation}
The \emph{quantization} of $D_{\xi}$ is then the \emph{field operator} on $\cH$ given by
\begin{equation}
    \wq{D_{\xi}}=\frac{1}{\sqrt{2}}(a_{J\xi}+a^{\dagger}_{J\xi}),
    \label{eq:linca}
\end{equation}
where $a^\dagger_{\phi}$ and $a_{\phi}$ denote the creation and annihilation operators associated to $\phi\in L$. The latter satisfy the commutation relations,
\begin{equation}
    [a_{\xi},a^{\dagger}_{\phi}]=\{\phi,\xi\} .
    \label{eq:ccr}
\end{equation}
We take a field operator to be any operator on $\cH$ that can be represented as the sum of a creation operator with the corresponding annihilation operator. Thus, we have one-to-one correspondences between elements of $L$, linear observables and field operators.
Note that a field operator is necessarily unbounded, having the real line as its spectrum.

Let $\cD\subseteq\cH$ be the dense subspace generated by the coherent states. We take this to be the domain of the field operators. We have,
\begin{equation}
    \langle \coh_{\gamma}, \wq{D_{\xi}}\, \coh_{\beta}\rangle
    = \langle \coh_{\gamma}, \coh_{\beta}\rangle D_{\xi}(\beta,\gamma),
    \label{eq:qlinobs}
\end{equation}
where $(\beta,\gamma)$ denotes the element in $L^\C$ given by,
\begin{equation}
    (\beta,\gamma)\defeq \frac12\left((1+\im J)\beta+(1-\im J)\gamma\right) .
    \label{eq:betagammaint}
\end{equation}
We have also extended $D_{\xi}$ to a complex linear function $L^\C\to\C$. Note,
\begin{equation}
    \|\wq{D_{\xi}}\, \coh_{\beta}\|_{\cH}^2=\|\coh_{\beta}\|_{\cH}^2 \left(\frac12\|\xi\|_{L}^2+\left(D_{\xi}(\beta)\right)^2\right).
\end{equation}

The standard quantization prescriptions all coincide for linear observables. For non-linear observables they may be characterized by the following identity \cite{Oe:feynobs}. Let $F_{\xi}=\exp(\im D_{\xi})$. Consider,
\begin{equation}
    \langle \coh_{\gamma}, \wq{F_{\xi}}\, \coh_{\beta}\rangle
    =\langle \coh_{\gamma},\coh_{\beta}\rangle
    F_{\xi}(\beta,\gamma) \exp\left(-c\,\|\xi\|^2\right) .
    \label{eq:quantid}
\end{equation}
For \emph{Weyl quantization} we have $c=1/4$. This is the most relevant quantization prescription for quantum field theory, due to its compatibility with the path integral prescription. For normal-ordered (Wick ordered) quantization, we have $c=0$. For anti-normal-ordered (anti-Wick ordered) quantization, we have $c=1/2$. We also say that $F_{\xi}$ is the Weyl/normal-ordered/anti-normal-ordered symbol of the operator $\wq{F_{\xi}}$. In the following, we are exclusively interested in Weyl quantization, the standard prescription in quantum field theory.

Since Weyl quantization preserves the exponential function, we can interpret equation~\eqref{eq:quantid} as the characteristic function of a matrix element of the spectral measure of the operator $\wq{D_{\xi}}$. Also replacing $\xi$ by $s\xi$, we rewrite it to this end as,
\begin{equation}
    \langle \coh_{\gamma}, \exp\left(\im s \wq{D_{\xi}}\right)\, \coh_{\beta}\rangle
    =\langle \coh_{\gamma},\coh_{\beta}\rangle
    \exp\left(\im s D_{\xi}(\beta,\gamma)\right) \exp\left(-\frac{s^2}{4}\|\xi\|^2\right) .
    \label{eq:charspec}
\end{equation}
Assuming the spectral measure to be absolutely continuous with respect to the Lebesgue measure and setting $\mu=D_{\xi}(\beta,\gamma)$ as well as $\sigma^2=\frac12 \|\xi\|^2$ we must have,
\begin{equation}
    \int_{-\infty}^{\infty}\xd q\, e^{\im s q} B_{\xi}(q)(\coh_{\gamma},\coh_{\beta})
    =\langle \coh_{\gamma},\coh_{\beta}\rangle
    \exp\left(\im s \mu -\frac12 \sigma^2 s^2\right) ,
\end{equation}
where $B_{\xi}(q):\cD\times\cD\to\C$ is our notation for the sesquilinear form determining the spectral measure. We recognize the characteristic function as the complex multiple of that of the Gaussian measure with mean $\mu$ and variance $\sigma^2$. That is,
\begin{equation}
    B_{\xi}(q)(\coh_{\gamma},\coh_{\beta})=\langle \coh_{\gamma},\coh_{\beta}\rangle \frac{1}{\sqrt{\pi}\|\xi\|}
    \exp\left(-\frac{1}{\|\xi\|^2}\left(D_{\xi}(\beta,\gamma)-q \right)^2\right) .
    \label{eq:smeasure}
\end{equation}


\section{POVM and spectral decompositions}
\label{sec:povmdec}

In the following we show rigorously that expression~\eqref{eq:smeasure} determines the spectral decomposition of $\wq{D_{\xi}}$. Moreover, we exhibit a family of positive-operator-valued measure (POVM) decompositions with the spectral decomposition as a limit. We denote by $\cB$ the space of bounded operators on $\cH$, equipped with the operator norm $\|\cdot\|_{\nop}$.

\subsection{Pointwise measure}

It is convenient to introduce a generalization of the form $B_\xi(q)$, involving an additional parameter $\epsilon\ge 0$ as follows,
\begin{align}
    B_{\xi}^{\epsilon}(q)(\coh_{\gamma},\coh_\beta)
    \defeq \langle \coh_{\gamma},\coh_{\beta}\rangle\,
    \frac{1}{\sqrt{\pi (\|\xi\|^2+\epsilon^2)}}
    \exp\left(-\frac{1}{\|\xi\|^2+\epsilon^2} \left(D_{\xi}(\beta,\gamma)-q\right)^2\right) .
    \label{eq:defbepsilon}
\end{align}

The previous form is recovered in the case $\epsilon=0$, $B_{\xi}(q)=B_{\xi}^{0}(q)$. Note that $B_{\xi}^{\epsilon}(q)$ is also well-defined if $\xi=0$, as long as $\epsilon>0$. In that case,
\begin{equation}
    B_{0}^{\epsilon}(q)(\coh_{\gamma},\coh_\beta)
    =\langle \coh_{\gamma},\coh_{\beta}\rangle \frac{e^{-q^2/\epsilon^2}}{\sqrt{\pi} \epsilon} .
\end{equation}
The following presentation implies that $B_{\xi}^{\epsilon}(q)$ is a hermitian sesquilinear form on $\cD$.

\begin{lem}
    \begin{multline}
        B_{\xi}^{\epsilon}(q)(\coh_{\gamma},\coh_\beta)
        = \frac{\|\xi\| e^{-\frac{q^2}{\|\xi\|^2+\epsilon^2}}}{\pi (\|\xi\|^2+\epsilon^2)}  \int_{-\infty}^{\infty}\xd r\, e^{-\frac{\|\xi\|^2 r^2}{\|\xi\|^2+\epsilon^2}}\, \int_{\xL}\xd\nu(\phi) \nonumber \\
        \langle\coh_{\gamma},\coh_{\phi+\frac{2 \|\xi\| r}{\|\xi\|^2+\epsilon^2}\xi+\frac{2 q}{\|\xi\|^2+\epsilon^2} J\xi}\rangle \langle\coh_{\phi-\frac{2 \|\xi\| r}{\|\xi\|^2+\epsilon^2}\xi+\frac{2 q}{\|\xi\|^2+\epsilon^2} J\xi},\coh_{\beta}\rangle \nonumber .
    \end{multline}
\end{lem}

We will later show that it is also positive. 

\begin{lem}
    \label{lem:weakcont}
    The hermitian sesquilinear forms $B_{\xi}^{\epsilon}(q)$ form a weakly continuous family in the parameters $(\xi,\epsilon,q)$. That is, for fixed $\psi,\eta\in\cD$, the function
    \begin{equation}
        (L\times\R_0^+\times\R)\setminus (\{0\}\times\{0\}\times\R)\to\C:
        (\xi,\epsilon,q)\mapsto B_{\xi}^{\epsilon}(q)(\eta,\psi)
    \end{equation}
    is continuous.
\end{lem}
\begin{proof}
    It is sufficient to consider the special case where $\psi$ and $\eta$ are coherent states. The result then follows from inspection of the defining expression (\ref{eq:defbepsilon}).
\end{proof}

We define the product $A'\star A$ of sesquilinear forms $A,A'$ on $\cD$ as follows, if the integral exists,
\begin{equation}
    (A'\star A)(\eta,\psi)\defeq \int_{\xL}\xd\nu(\phi)\, A'(\eta,\coh_\phi) A(\coh_\phi,\psi) .
\end{equation}

\begin{lem}
    \label{lem:prod}
    If either $\xi,\xi'$ are real linearly independent or $\epsilon>0$ or $\epsilon'>0$, then,
    \begin{multline}
        (B_{\xi'}^{\epsilon'}(q')\star B_{\xi}^{\epsilon}(q))(\coh_\gamma, \coh_\beta)
        =\langle \coh_\gamma, \coh_\beta\rangle
        \frac{1}{\pi\sqrt{(\|\xi\|^2+\epsilon^2) (\|\xi'\|^2+\epsilon'^2)- \{\xi,\xi'\}^2}}\\
        \exp\left(\frac{1}{(\|\xi\|^2+\epsilon^2) (\|\xi'\|^2+\epsilon'^2)- \{\xi,\xi'\}^2}
           \left(
           2\{\xi,\xi'\}\left(D_{\xi}(\beta,\gamma)-q\right)\left(D_{\xi'}(\beta,\gamma)-q'\right)\right. \right. \\
           \left. \left. -(\|\xi'\|^2+\epsilon'^2)\left(D_{\xi}(\beta,\gamma)-q\right)^2
           -(\|\xi\|^2+\epsilon^2)\left(D_{\xi'}(\beta,\gamma)-q'\right)^2
       \right)\right) .
    \label{eq:prod}
    \end{multline}
    In particular, if $\omega(\xi,\xi')=0$, then $B_{\xi'}^{\epsilon'}(q')\star B_{\xi}^{\epsilon}(q)=B_{\xi}^{\epsilon}(q)\star B_{\xi'}^{\epsilon'}(q')$.
\end{lem}
\begin{proof}
    Equation \eqref{eq:prod} is obtained by explicit calculation, see Appendix~\ref{sec:app}. It is valid as long as the denominators in the fractions are not zero. By the Cauchy-Schwarz inequality these are zero precisely if $\epsilon=\epsilon'=0$ and moreover $\xi'$ and $\xi$ are linearly dependent over $\R$, so that $\{\xi,\xi'\}=\|\xi\| \|\xi'\|$. If $\omega(\xi,\xi')=0$, then $\{\xi,\xi'\}=\{\xi',\xi\}$ so that the right-hand side of equation \eqref{eq:prod} becomes symmetric under interchange of $(\xi,\epsilon,q)$ and $(\xi',\epsilon',q')$.
\end{proof}

\begin{lem}
    \label{lem:projsumid}
    If either $\xi,\xi'$ are real linearly independent or $\epsilon>0$ or $\epsilon'>0$, then,
    \begin{multline}
        \int_{-\infty}^{\infty}\xd s\, \left(B_{\xi'}^{\epsilon'}(q-s) \star B_{\xi}^{\epsilon}(s)\right) (\coh_\gamma,\coh_\beta)\\
         =\langle \coh_\gamma, \coh_\beta\rangle
         \frac{1}{\sqrt{\pi}\sqrt{2\{\xi,\xi'\}+\|\xi\|^2+\|\xi'\|^2+\epsilon^2+\epsilon'^2}}\\
         \exp\left(-\frac{1}{2\{\xi,\xi'\}+\|\xi\|^2+\|\xi'\|^2+\epsilon^2+\epsilon'^2}
         \left(D_{\xi+\xi'}(\beta,\gamma)-q\right)^2\right) .
         \label{eq:projsumid}
    \end{multline}
    If, moreover, $\omega(\xi,\xi')=0$, then,
    \begin{equation}
        \int_{-\infty}^{\infty}\xd s\, \left(B_{\xi'}^{\epsilon'}(q-s) \star B_{\xi}^{\epsilon}(s)\right) (\coh_\gamma,\coh_\beta)
        =B_{\xi+\xi'}^{\epsilon''}(q)(\coh_\gamma,\coh_\beta) ,
        \label{eq:projcomid}
    \end{equation}
    where $\epsilon''^2=\epsilon^2 +\epsilon'^2$.
\end{lem}
\begin{proof}
    Equation \eqref{eq:projsumid} is obtained by explicit calculation from Lemma~\ref{lem:prod}, see Appendix~\ref{sec:app}. If $\omega(\xi,\xi')=0$, the right-hand side is recognized to take the form specified in equation \eqref{eq:projcomid}.
\end{proof}

\begin{lem}
    \label{lem:comp}
    If either $\epsilon>0$ or $\epsilon'>0$, then,
    \begin{equation}
    B_{\xi}^{\epsilon'}(q')\star B_{\xi}^{\epsilon}(q)
    = \frac{1}{\sqrt{\pi (\epsilon^2+\epsilon'^2)}}
    \exp\left(-\frac{1}{\epsilon^2+\epsilon'^2}(q-q')^2\right)
    B_{\xi}^{\epsilon''}\left(\frac{\epsilon'^2 q +\epsilon^2 q'}{\epsilon^2+\epsilon'^2}\right),
    \quad\text{with}\;\; \epsilon''^2=\frac{\epsilon^2 \epsilon'^2}{\epsilon^2+\epsilon'^2} .
    \end{equation}
\end{lem}
\begin{proof}
    This arises from a special case of Lemma~\ref{lem:prod}, recognizing the resulting expression as a hermitian sesquilinear form of the same family.
\end{proof}

\begin{lem}
    \label{lem:pos}
    The hermitian sesquilinear forms $B_{\xi}^{\epsilon}(q)$ are positive.
\end{lem}
\begin{proof}
    If $\epsilon>0$ then, from Lemma~\ref{lem:comp} we have that
    \begin{equation}
        B_{\xi}^{\epsilon}(q)=2\sqrt{\pi}\epsilon\,  B_{\xi}^{\sqrt2\epsilon}(q)\star B_{\xi}^{\sqrt2\epsilon}(q) .
    \end{equation}
    That is, $B_{\xi}^{\epsilon}(q)$ is a positive multiple of a square of a hermitian sesquilinear form and is thus positive. Positivity of $B_{\xi}^{0}(q)$ follows by taking the limit $\epsilon\to 0$ of $B_{\xi}^{\epsilon}(q)(\psi,\psi)$ for $\psi\in\cD$ due to Lemma~\ref{lem:weakcont}.
\end{proof}

\begin{lem}
    \label{lem:unitaryintegral}
    For $\epsilon>0$ we have,
    \begin{equation}
        B_{\xi}^{\epsilon}(q)(\coh_{\gamma},\coh_{\beta})
        =\frac{1}{\pi}\,
        \int_{-\infty}^{\infty}\xd t\, e^{-\epsilon^2 t^2-2\im q t}
        \langle \coh_{\gamma},\wq{G_{\xi}^t}\coh_{\beta}\rangle .
        \label{eq:unitaryintegral}
    \end{equation}
    The operator $\wq{G_{\xi}^t}$ is unitary on $\cH$ and its Weyl symbol is given by,
    \begin{equation}
        G_{\xi}^t(\phi)\defeq \exp\left(2\im t (D_{\xi}(\phi)\right) .
        \label{eq:unitaryg}
    \end{equation}
\end{lem}
\begin{proof}
    As is easy to see, if $\epsilon>0$, the definition~\eqref{eq:defbepsilon} of $B_{\xi}^{\epsilon}(q)$ is equivalent to
    \begin{equation}
        B_{\xi}^{\epsilon}(q)(\coh_{\gamma},\coh_{\beta})
        =\frac{1}{\pi}\, \int_{-\infty}^{\infty}\xd t\, e^{-\epsilon^2 t^2-2\im q t}
         \langle \coh_{\gamma},\coh_{\beta}\rangle\,
         \exp\left(2\im t D_{\xi}(\beta,\gamma)-\frac{1}{4}(2 t \|\xi\|)^2\right) .
        \label{eq:intrepsql}        
    \end{equation}
    By comparison with expression \eqref{eq:charspec} we recognize the integrand (without the first exponential factor and with $s=2 t$) as the Weyl quantization of the observable \eqref{eq:unitaryg}.
\end{proof}

\begin{lem}
    \label{lem:bopq}
    For $\epsilon>0$ the positive hermitian sesquilinear form $B_{\xi}^\epsilon(q)$ on $\cD$ extends to a positive hermitian sesquilinear form on $\cH$. Moreover, there exists a unique bounded positive operator $\Pi_{\xi}^\epsilon(q)$ such that for any $\eta,\psi\in\cH$,
    \begin{equation}
        B_{\xi}^\epsilon(q)(\eta,\psi)=\langle\eta,\Pi_{\xi}^\epsilon(q)\psi\rangle,\quad\text{with}\quad
        \|\Pi_{\xi}^\epsilon(q)\|_{\nop}\le \frac{1}{\sqrt{\pi} \epsilon} .
    \end{equation}
\end{lem}
\begin{proof}
    Since $\wq{G_{\xi}^t(q)}$ of Lemma~\ref{lem:unitaryintegral} is a unitary operator, we have from the Cauchy-Schwarz inequality $|\langle\psi,\wq{G_{\xi}^t(q)}\psi\rangle|\le \langle\psi,\psi\rangle$ for $\psi\in\cH$. This yields from equation~\eqref{eq:unitaryintegral} for $\psi\in\cD$ the bound,
    \begin{equation}
        B_{\xi}^\epsilon(q)(\psi,\psi)
        \le \frac{1}{\pi}\, \int_{-\infty}^{\infty}\xd t\, e^{-\epsilon^2 t^2} \langle\psi,\psi\rangle
        =\frac{1}{\sqrt{\pi}\epsilon} \langle\psi,\psi\rangle .
    \end{equation}
    This implies that $B_{\xi}^\epsilon(q)(\psi,\psi)$ extends to a positive hermitian sesquilinear form on $\cH$ with the same bound. Moreover, there exists thus a unique bounded positive operator $\Pi_{\xi}^{\epsilon}(q)$ on $\cH$, related to $B_{\xi}^\epsilon(q)(\psi,\psi)$ as stated. The inequality
    \begin{equation}
        \langle\psi,\Pi_{\xi}^\epsilon(q) \psi\rangle \le
        \frac{1}{\sqrt{\pi}\epsilon} \langle\psi,\psi\rangle
    \end{equation}
    for $\psi\in\cH$ then implies the stated bound on $\Pi_{\xi}^{\epsilon}(q)$.
\end{proof}

\begin{lem}
    \label{lem:strongcont}
    For $\epsilon>0$ the positive operators $\Pi_{\xi}^{\epsilon}(q)$ form a strongly continuous family in the parameters $(\xi,\epsilon,q)$. That is, for fixed $\psi\in\cH$, the function
    \begin{equation}
        L\times\R^+\times\R \to\cH:
        (\xi,\epsilon,q)\mapsto \Pi_{\xi}^{\epsilon}(q) \psi
    \end{equation}
    is continuous.
\end{lem}
\begin{proof}
    Due to the boundedness of $\Pi_{\xi}^{\epsilon}(q)$ it is sufficient to consider $\psi$ to be a coherent state, say $\coh_{\beta}$. Using self-adjointness, we then write,
    \begin{multline}
        \left\|\left(\Pi_{\xi'}^{\epsilon'}(q')-\Pi_{\xi}^{\epsilon}(q)\right)\coh_{\beta}\right\|^2 \\
        =\left(B_{\xi'}^{\epsilon'}(q')\star B_{\xi'}^{\epsilon'}(q')
        -B_{\xi}^{\epsilon}(q)\star B_{\xi'}^{\epsilon'}(q')
        -B_{\xi'}^{\epsilon'}(q')\star B_{\xi}^{\epsilon}(q)
        +B_{\xi}^{\epsilon}(q)\star B_{\xi}^{\epsilon}(q)\right)
        (\coh_{\beta},\coh_{\beta}) .
    \end{multline}
    Continuity can be read off from the explicit expressions obtained with Lemma~\ref{lem:prod}.
\end{proof}

\begin{lem}
    \label{lem:normcont}
    For $\epsilon>0$ the positive operators $\Pi_{\xi}^{\epsilon}(q)$ form norm continuous families in the parameters $(\epsilon,q)$. That is, the function $\R^+\times\R\to\cB$ given by $(\epsilon,q)\mapsto \Pi_{\xi}^{\epsilon}(q)$ is continuous.
\end{lem}
\begin{proof}
    We use the representation of equation \eqref{eq:unitaryintegral} of Lemma~\ref{lem:unitaryintegral} which by Lemma~\ref{lem:bopq} extends to arbitrary matrix elements. Thus, we have the estimate,
    \begin{equation}
        \left\|\Pi_{\xi}^{\epsilon}(q)-\Pi_{\xi}^{\epsilon}(q')\right\|_{\nop}
        \le \frac{1}{\pi}\int_{-\infty}^{\infty}\xd t\, e^{-\epsilon^2 t^2}
        \left|e^{-2\im q t}-e^{-2\im q' t}\right|
        \left\|\wq{G_{\xi}^t}\right\|_{\nop}
        \le |q-q'| \frac{1}{\pi}\int_{-\infty}^{\infty} e^{-\epsilon^2 t^2}
        2 |t| .
    \end{equation}
    Note that the integral is finite and does not depend on $q$ or $q'$. On the other hand, suppose $\epsilon'\le\epsilon$. Then,
    \begin{equation}
        \left\|\Pi_{\xi}^{\epsilon}(q')-\Pi_{\xi}^{\epsilon'}(q')\right\|_{\nop}
        \le \frac{1}{\pi}\int_{-\infty}^{\infty}\xd t\, \left(e^{-\epsilon'^2 t^2}-e^{-\epsilon^2 t^2}\right)
        \left|e^{-2\im q' t}\right|
        \left\|\wq{G_{\xi}^t}\right\|_{\nop}
        = \frac{1}{\sqrt{\pi}}\left(\frac{1}{\epsilon'}-\frac{1}{\epsilon}\right) .
        \label{eq:piqepsilon}
    \end{equation}
    Combining, we obtain a suitable estimate for,
    \begin{equation}
        \left\|\Pi_{\xi}^{\epsilon}(q)-\Pi_{\xi}^{\epsilon'}(q')\right\|_{\nop}
        \le \left\|\Pi_{\xi}^{\epsilon}(q)-\Pi_{\xi}^{\epsilon}(q')\right\|_{\nop} + \left\|\Pi_{\xi}^{\epsilon}(q')-\Pi_{\xi}^{\epsilon'}(q')\right\|_{\nop} .
    \end{equation}
    This shows the desired continuity.
\end{proof}

\begin{prop}
    \label{prop:piobs}
    For $\epsilon>0$, $\Pi_{\xi}^{\epsilon}(q)$ is the Weyl quantization of the observable
    \begin{equation}
       H_{\xi}^{\epsilon}(q)(\phi)
       \defeq\frac{1}{\sqrt{\pi}\epsilon}\exp\left(-\frac{1}{\epsilon^2}(D_{\xi}(\phi)-q)^2\right) .
    \end{equation}
\end{prop}
\begin{proof}
    We note, with definition \eqref{eq:unitaryg},
    \begin{equation}
       H_{\xi}^{\epsilon}(q)(\phi)
       =\frac{1}{\pi}\, \int_{-\infty}^{\infty}\xd t\, e^{-\epsilon^2 t^2-2\im q t} G_{\xi}^t(\phi) .
    \end{equation}
    Applying Weyl quantization on both sides, linearity (and integrability) on the right-hand side imply that it is equivalent to quantize just $G_{\xi}^t$. By Lemma~\ref{lem:unitaryintegral} the right-hand side is then equal to $\Pi_{\xi}^{\epsilon}(q)$.
\end{proof}

\subsection{Functional calculus}

We introduce the following notation, where $\psi,\eta\in\cD$, and $f:\R\to\C$ is a Lebesgue measurable function, whenever the integral exists,
\begin{equation}
    B_{\xi}^{\epsilon}[f](\psi,\eta)\defeq \int_{-\infty}^{\infty}\xd q f(q) B_{\xi}^{\epsilon}(q)(\psi,\eta) .
    \label{eq:defbint}
\end{equation}
We note that $q\mapsto B_{\xi}^{\epsilon}(q)(\psi,\eta)$ is continuous by Lemma~\ref{lem:weakcont} and thus in particular measurable.
If $\epsilon>0$ we also define, whenever the integral exists,
\begin{equation}
    \Pi_{\xi}^{\epsilon}[f] \defeq \int_{-\infty}^{\infty}\xd q f(q) \Pi_{\xi}^{\epsilon}(q) .
    \label{eq:defpint}
\end{equation}
For this to make sense, we use continuity of the map $q\mapsto \Pi_{\xi}^{\epsilon}(q)$, due to Lemma~\ref{lem:normcont}. This implies here measurability and also approximability by simple functions on every closed finite interval. We also consider the weaker version of definition~\eqref{eq:defpint} in the strong operator topology.

\begin{lem}
    \label{lem:integrals}
    Let $\one(q)=1$, $\id(q)=q$, $e_s(q)=e^{\im s q}$. Then,
    \begin{align}
        B_{\xi}^{\epsilon}[\one](\coh_\gamma,\coh_\beta)
         & =\langle\coh_\gamma, \coh_\beta\rangle, \label{eq:intconst} \\
        B_{\xi}^{\epsilon}[\id](\coh_\gamma,\coh_\beta)
         & =\langle\coh_\gamma, \coh_\beta\rangle D_{\xi}(\beta,\gamma)
        =\langle\coh_\gamma, \wq{D_{\xi}}\coh_\beta\rangle, \label{eq:intlin} \\
        B_{\xi}^{\epsilon}[e_s](\coh_\gamma,\coh_\beta)
         & = \langle\coh_\gamma, \coh_\beta\rangle \exp\left(\im s D_{\xi}(\beta,\gamma)\right) \exp\left(-\frac{s^2}{4} (\|\xi\|^2+\epsilon^2)
         \right) .
    \end{align}
\end{lem}
\begin{proof}
    This is obtained by direct computation with expression \eqref{eq:defbepsilon}. The second equality in \eqref{eq:intlin} is equation \eqref{eq:qlinobs}.
\end{proof}

\begin{lem}
    \label{lem:fbound}
    Let $f:\R\to\C$ be Lebesgue measurable and essentially bounded. Then, $B_{\xi}^{\epsilon}[f]$ extends to a sesquilinear form on $\cH$. Moreover, there is a unique bounded operator $\Pi_{\xi}^{\epsilon}[f]$ on $\cH$ such that for any $\eta,\psi\in\cH$,
    \begin{equation}
        B_{\xi}^{\epsilon}[f](\eta,\psi)=\langle\eta,\Pi_{\xi}^{\epsilon}[f]\psi\rangle,\quad\text{with}\quad \|\Pi_{\xi}^{\epsilon}[f]\|_{\nop}\le \|\Re(f)\|_{\infty} +\|\Im(f)\|_{\infty} .
        \label{eq:formop}
    \end{equation}
    Also, $\Pi_{\xi}^{\epsilon}[\overline{f}]=\Pi_{\xi}^{\epsilon}[f]^{\dagger}$. Furthermore, if $f$ is real, then $B_{\xi}^{\epsilon}[f]$ is hermitian and $\Pi_{\xi}^{\epsilon}[f]$ is self-adjoint. Also, if $f$ is positive, i.e.\ $f\ge 0$, then $B_{\xi}^{\epsilon}[f]$ is positive and $\Pi_{\xi}^{\epsilon}[f]$ is positive.
\end{lem}
\begin{proof}
    Assume at first that $f$ is positive.
    From the definition \eqref{eq:defbint} and Lemma~\ref{lem:pos} it is clear that $B_{\xi}^{\epsilon}[f]$ is a positive hermitian sesquilinear form on $\cD$. Without loss of generality suppose that $f$ is bounded and $\|f\|_{\infty}=\|f\|_{\mathrm{sup}}$. For $\psi\in\cD$ we obtain with equation \eqref{eq:intconst} of Lemma~\ref{lem:integrals},
    \begin{equation}
        0\le B_{\xi}^{\epsilon}[f](\psi,\psi)\le B_{\xi}^{\epsilon}[\|f\|_{\infty} \one](\psi,\psi)= \|f\|_{\infty} B_{\xi}^{\epsilon}[\one](\psi,\psi) = \|f\|_{\infty} \langle\psi,\psi\rangle .
    \end{equation}
    That is, the positive hermitian sesquilinear form $B_{\xi}^{\epsilon}[f]$ is bounded on the dense subspace $\cD$ by the inner product with a constant. It therefore extends to a positive hermitian sesquilinear form on $\cH$, with the same bound. In turn, it determines a bounded positive operator $\Pi_{\xi}^{\epsilon}[f]$ satisfying the equation in \eqref{eq:formop}.

    Let $f$ be complex-valued. By decomposing $f$ into real and imaginary part, and each of these into positive and negative part we can see that $B_{\xi}^{\epsilon}[f]$ is still a well-defined sesquilinear form on $\cH$ and $\Pi_{\xi}^{\epsilon}[f]$ is still a bounded operator, related by the equation in \eqref{eq:formop}. Moreover, $B_{\xi}^{\epsilon}[f]$ is hermitian and $\Pi_{\xi}^{\epsilon}[f]$ is self-adjoint if $f$ is real-valued. Also, complex conjugation of $f$ corresponds to taking the adjoint of $\Pi_{\xi}^{\epsilon}[f]$.
    
    For $\psi\in\cH$ we have,
    \begin{equation}
        0\le |\langle\psi,\Pi_{\xi}^{\epsilon}[f]\psi\rangle |
        \le \langle\psi,\Pi_{\xi}^{\epsilon}[|f|]\psi\rangle
        \le \|f\|_{\infty} B_{\xi}^{\epsilon}[\one](\psi,\psi)
        \le \|f\|_{\infty}\langle\psi,\psi\rangle.
    \end{equation}
    If $f$ is real-valued and $\Pi_{\xi}^{\epsilon}[f]$ therefore self-adjoint, this implies $\|\Pi_{\xi}^{\epsilon}[f]\|_{\nop}\le \|f\|_{\infty}$. If $f$ is complex-valued, we decompose it into its real and imaginary part, yielding the bound in expression \eqref{eq:formop}.
\end{proof}

\begin{lem}
    \label{lem:weakcontint}
    Let $f:\R\to\C$ be essentially bounded. Then, the family $\Pi_{\xi}^{\epsilon}[f]$ is weakly continuous in the parameters $(\xi,\epsilon)$. That is, for fixed $\psi,\eta\in\cH$, the function
    \begin{equation}
        (L\times\R_0^+)\setminus (\{0\}\times\{0\})\to\C:
        (\xi,\epsilon)\mapsto \langle\eta,\Pi_{\xi}^{\epsilon}[f]\psi\rangle
    \end{equation}
    is continuous.
\end{lem}
\begin{proof}
    For $\psi,\eta\in\cD$ the integral
    \begin{equation}
      \int_{-\infty}^\infty\xd q\,
      \left| B_{\xi'}^{\epsilon'}(q)(\eta,\psi)- B_{\xi}^{\epsilon}(q)(\eta,\psi)\right|
    \end{equation}
    converges to $0$ when $(\xi',\epsilon')\to (\xi,\epsilon)$. This implies the required weak continuity in this special case. The extension of weak continuity to $\psi,\eta\in\cH$ follows from the boundedness of $\Pi_{\xi}^{\epsilon}[f]$ with the denseness of $\cD\subseteq\cH$.
\end{proof}

\begin{lem}
    \label{lem:strongcontint}
    Let $\epsilon>0$, $f:\R\to\C$ essentially bounded and of compact support. Then, the family $\Pi_{\xi}^{\epsilon}[f]$ is strongly continuous in the parameters $(\xi,\epsilon)$. That is, for fixed $\psi\in\cH$, the function
    \begin{equation}
        L\times\R^+\to\cH:
        (\xi,\epsilon)\mapsto \Pi_{\xi}^{\epsilon}[f]\psi
    \end{equation}
    is continuous. Also, we have norm continuity in $\epsilon$. That is, the map $\R^+\to \cB:\epsilon\mapsto\Pi_{\xi}^{\epsilon}[f]$ is continuous.
\end{lem}
\begin{proof}
    Strong continuity follows from the strong continuity result of Lemma~\ref{lem:strongcont} with definition~\eqref{eq:defpint}. Norm continuity follows with Lemma~\ref{lem:normcont}.
\end{proof}

\begin{lem}
    \label{lem:dcommute}
    Let $\xi,\xi'\in L$ with $\omega(\xi,\xi')=0$. Then, $[\Pi_{\xi'}^{\epsilon'}[f'],\Pi_{\xi}^{\epsilon}[f]]=0$.
\end{lem}
\begin{proof}
    In the context of Lemma~\ref{lem:fbound} this follows from Lemma~\ref{lem:prod}. In the special case that $\xi,\xi'$ are linearly dependent and $\epsilon=\epsilon'=0$ we can set $\epsilon'>0$ and take the weak limit $\epsilon'\to 0$ using Lemma~\ref{lem:weakcontint}. (The weak limit is sufficient here as we already know of the existence of the limiting operator and just have to check that its matrix elements vanish.) 
\end{proof}

Denote by $\chi_A$ the characteristic function of a set $A$.
Denote by $\cM$ be the Borel $\sigma$-algebra of $\R$.

\begin{lem}
    \label{lem:intproj}
    Let $A,B\in\cM$. Then,
    \begin{equation}
    \Pi_{\xi}[\chi_A] \Pi_{\xi}[\chi_B] = \Pi_{\xi}[\chi_{A\cap B}] .
    \end{equation}
\end{lem}
\begin{proof}
    Due to boundedness of the operators it is enough to show this for matrix elements in $\cD$. Moreover, it is enough to consider finite intervals for $A,B$ as they generate the Borel $\sigma$-algebra on $\R$, using positivity, additivity and the bound $0\le\Pi_{\xi}[\chi_X]\le\id$ for any $X\in\cM$.
    
    We first consider a finite interval $[a,b]$. Then, using Lemmas~\ref{lem:weakcontint} and \ref{lem:comp}, 
    \begin{align}
        \langle\coh_{\gamma},\Pi_{\xi}[\chi_{[a,b]}]\Pi_{\xi}[\chi_{[a,b]}]\coh_{\beta}\rangle
        & =\lim_{\epsilon\to 0} \langle\coh_{\gamma},\Pi_{\xi}[\chi_{[a,b]}]\Pi_{\xi}^{\epsilon}[\chi_{[a,b]}]\coh_{\beta}\rangle \\
        & =\lim_{\epsilon\to 0} B_{\xi}[\chi_{[a,b]}]\star B_{\xi}^{\epsilon}[\chi_{[a,b]}](\coh_{\gamma},\coh_{\beta}) \\
        & =\lim_{\epsilon\to 0} \int_{a}^{b} \xd q' \int_{a}^{b} \xd q\,
        B_{\xi}(q')\star B_{\xi}^{\epsilon}(q)(\coh_{\gamma},\coh_{\beta}) \\
        & =\lim_{\epsilon\to 0} \int_{a}^{b} \xd q' \int_{a}^{b} \xd q\,
        \frac{1}{\sqrt{\pi}\, \epsilon}\exp\left(-\frac{1}{\epsilon^2}(q-q')^2\right) B_{\xi}(q')(\coh_{\gamma},\coh_{\beta}) \\
        & =\int_{a}^{b} \xd q'\,  B_{\xi}(q')(\coh_{\gamma},\coh_{\beta})\lim_{\epsilon\to 0} \int_{a}^{b} \xd q\,
        \frac{1}{\sqrt{\pi}\, \epsilon}\exp\left(-\frac{1}{\epsilon^2}(q-q')^2\right) \\
        & =\int_{a}^{b} \xd q'\,  B_{\xi}(q')(\coh_{\gamma},\coh_{\beta})\lim_{\epsilon\to 0} \int_{-\infty}^{\infty} \xd q\,
        \frac{1}{\sqrt{\pi}\, \epsilon}\exp\left(-\frac{1}{\epsilon^2}(q-q')^2\right) \\
        & =\int_{a}^{b} \xd q'\,  B_{\xi}(q')(\coh_{\gamma},\coh_{\beta})\lim_{\epsilon\to 0} \int_{-\infty}^{\infty} \xd q\,
        \frac{1}{\sqrt{\pi}\, \epsilon}\exp\left(-\frac{1}{\epsilon^2}q^2\right) \\
        & =\int_{a}^{b} \xd q'\,  B_{\xi}(q')(\coh_{\gamma},\coh_{\beta})
        = \langle\coh_{\gamma},\Pi_{\xi}[\chi_{[a,b]}]\coh_{\beta}\rangle .
    \end{align}
    With this, we have established that $\Pi_{\xi}[\chi_{[a,b]}]$ is an orthogonal projection. Suppose now that $A,B\in\cM$ are disjoint finite intervals. Then, using commutativity (Lemma~\ref{lem:dcommute}),
    \begin{multline}
        \Pi_{\xi}[\chi_{A\cup B}]
        =(\Pi_{\xi}[\chi_{A\cup B}])^2
        =(\Pi_{\xi}[\chi_{A}]+\Pi_{\xi}[\chi_{B}])^2 \\
        =\Pi_{\xi}[\chi_{A}]+\Pi_{\xi}[\chi_{B}]+2 \Pi_{\xi}[\chi_{A}]\Pi_{\xi}[\chi_{B}]
        =\Pi_{\xi}[\chi_{A\cup B}]+2 \Pi_{\xi}[\chi_{A}]\Pi_{\xi}[\chi_{B}] .
    \end{multline}
    In particular, $\Pi_{\xi}[\chi_{A}]\Pi_{\xi}[\chi_{B}]=0$. Using that finite intervals generate the Borel $\sigma$-algebra on $\R$ we can complete the proof.
\end{proof}

\begin{thm}
    \label{thm:povmspecdec}
    The map $\cM\to B(\cH)$ given by $A\mapsto \Pi_{\xi}^{\epsilon}[\chi_A]$ is a positive-operator-valued measure (POVM), providing a decomposition of the operator $\wq{D_{\xi}}$ in the sense,
    \begin{equation}
        \langle\eta,\Pi_{\xi}^{\epsilon}[\id]\psi\rangle\defeq B_{\xi}^{\epsilon}[\id](\eta,\psi) = \langle\eta,\wq{D_{\xi}}\psi\rangle\qquad\forall\eta,\psi\in\cD .
        \label{eq:decomp}
    \end{equation}
    If $\epsilon=0$ this map is even a spectral measure, i.e.\ a projection valued measure (PVM). It provides the spectral decomposition of $\wq{D_{\xi}}$.
\end{thm}
\begin{proof}
    We use Lemma~\ref{lem:fbound} throughout. In particular, it implies positivity of $\Pi_{\xi}^{\epsilon}[\chi_A]$ for $A\in\cM$. Identity \eqref{eq:intlin} of Lemma~\ref{lem:integrals} yields expression~\eqref{eq:decomp}. By identity \eqref{eq:intconst} of Lemma~\ref{lem:integrals} we have $\Pi_{\xi}^{\epsilon}[\chi_{\R}]=\id_{\cH}$. Also, $\Pi_{\xi}^{\epsilon}[\chi_{\emptyset}]=0$ is clear. If $\epsilon>0$, strong  $\sigma$-additivity follows from definition~\eqref{eq:defpint}.
    Otherwise, we note that $\sigma$-additivity
    for matrix elements in $\cD$ follows from definition~\eqref{eq:defbint}. This readily implies finite additivity in $\cH$. To show weak $\sigma$-additivity in $\cH$ it is sufficient to consider matrix elements of the form $B_{\xi}^{\epsilon}[\chi_A](\psi,\psi)$ for $\psi\in\cH$. Fix $\delta>0$ and let $\eta\in\cD$ such that $(\|\psi\|+\|\eta\|) \|\psi-\eta\|<\delta$. Let $\{A_n\}_{n\in\N}$ be a sequence of disjoint elements in $\cM$ with $A=\cup_{n\in\N} A_n$.
    
    Note that for any bounded operator $B$ we have,
    \begin{equation}
        \left|\langle \psi,B\psi\rangle-\langle \eta,B\eta\rangle\right|
        =\left|\langle \psi,B(\psi-\eta)\rangle+\langle \psi-\eta,B\eta\rangle\right|
        \le \| B\|_{\Lop} (\|\psi\|+\|\eta\|) \|\psi-\eta\| .
    \end{equation}
    By Lemma~\ref{lem:fbound}, $\|\Pi_{\xi}^{\epsilon}[\chi_A]\|_{\Lop}\le \|\chi_A\|_{\infty}=1$ and thus,
    \begin{equation}
        \left|B_{\xi}^{\epsilon}[\chi_A](\psi,\psi)-B_{\xi}^{\epsilon}[\chi_A](\eta,\eta)\right|<\delta .
        \label{eq:sd1}
    \end{equation}
    By $\sigma$-additivity in $\cD$ we have,
    \begin{equation}
        B_{\xi}^{\epsilon}[\chi_A](\eta,\eta)=\sum_{n=1}^{\infty} B_{\xi}^{\epsilon}[\chi_{A_n}](\eta,\eta) .
    \end{equation}
    By positivity of the summands there is $k\in\N$ such that
    \begin{equation}
        \left|B_{\xi}^{\epsilon}[\chi_A](\eta,\eta)-\sum_{n=1}^{k} B_{\xi}^{\epsilon}[\chi_{A_n}](\eta,\eta)\right|< \delta.
        \label{eq:sd2}
    \end{equation}
    On the other hand, using finite additivity and the same type of estimate as \eqref{eq:sd1} we have,
    \begin{equation}
        \left|\sum_{n=1}^{k} B_{\xi}^{\epsilon}[\chi_{A_n}](\psi,\psi)-\sum_{n=1}^{k} B_{\xi}^{\epsilon}[\chi_{A_n}](\eta,\eta)\right|< \delta.
        \label{eq:sd3}
    \end{equation}
    Combining \eqref{eq:sd1}, \eqref{eq:sd2} and \eqref{eq:sd3} yields the inequality,
    \begin{equation}
        \left|B_{\xi}^{\epsilon}[\chi_{A}](\psi,\psi)-\sum_{n=1}^{k} B_{\xi}^{\epsilon}[\chi_{A_n}](\psi,\psi)\right|< 3\delta.
    \end{equation}
    Since the difference on the left-hand side is non-negative, this implies convergence and thus the required weak $\sigma$-additivity. Now strong $\sigma$-additivity follows from weak $\sigma$-additivity as follows. Let $\psi\in\cH$ and consider,
    \begin{equation}
        \left\| \left(\Pi_{\xi}^{\epsilon}[\chi_A]-\sum_{n=1}^{\infty}\Pi_{\xi}^{\epsilon}[\chi_{A_n}]\right)\psi\right\|^2
        =\left\langle\psi,\left(\Pi_{\xi}^{\epsilon}[\chi_A]-\sum_{n=1}^{\infty}\Pi_{\xi}^{\epsilon}[\chi_{A_n}]\right)\left(\Pi_{\xi}^{\epsilon}[\chi_A]-\sum_{n=1}^{\infty}\Pi_{\xi}^{\epsilon}[\chi_{A_n}]\right)\psi\right\rangle
    \end{equation}
    Note that $\Pi_{\xi}^{\epsilon}[\chi_A]-\sum_{n=1}^{\infty}\Pi_{\xi}^{\epsilon}[\chi_{A_n}]$ is positive due to finite additivity and bounded by the identity. It is thus a bounded operator. We can then carry out the inner sum by weak $\sigma$-additivity, and we obtain zero.
\end{proof}


\section{Quantum operations}
\label{sec:quantop}

\subsection{Discrete outcomes}
\label{sec:doutcomes}

In non-relativistic quantum mechanics, we use the spectral decomposition of a self-adjoint operator to construct the quantum operations that encode the associated measurements. Given the spectral decomposition of the operator $\wq{D_\xi}$, there is no impediment to proceed in an analogous way in quantum field theory, as long as we are interested in discrete outcomes.

Denote by $\cT\subseteq\cB$ the space of trace-class operators on $\cH$ with the trace norm denoted by $\|\cdot\|_{\ntr}$. Let $X\subseteq\R$ be a Lebesgue measurable subset. Define the map $Q_{\xi}[X]:\cT\to\cT$ by,
\begin{equation}
    Q_{\xi}[X](\sigma)\defeq \Pi_{\xi}[\chi_{X}]\, \sigma\, \Pi_{\xi}[\chi_{X}] .
    \label{eq:projop}
\end{equation}
Evidently, $Q_{\xi}[X]$ is completely positive. $Q_{\xi}[X]$ encodes a \emph{selective quantum operation} that tests whether a measurement of the observable $D_{\xi}$ yields an outcome in the subset $X\subseteq\R$. In order to obtain a complete set of alternative outcomes we consider a countable partition of $\R$, say $\{X_k\}_{k\in I}$. That is, $X_k\subseteq\R$ is Lebesgue measurable, $X_k\cap X_j=\emptyset$ if $k\neq j$ and $\R=\bigcup_{k\in I} X_k$. Define $Q_{\xi}:\cT\to\cT$ by $Q_{\xi}\defeq\sum_{k\in I} Q_{\xi}[X_k]$. Then, $Q_{\xi}$ is completely positive and trace-preserving. $Q_{\xi}$ encodes the \emph{non-selective quantum operation} that corresponds to performing the measurement, without reading any outcome. The measurement is projective and thus \emph{ideal} in the sense $Q_{\xi}[X] \circ Q_{\xi}[X] = Q_{\xi}[X]$ and $Q_{\xi}\circ Q_{\xi}= Q_{\xi}$.

The type of measurement considered here extracts from the quantum system only the coarse grained information about which subset carries the outcome. Any two different partitions define different measurements and in particular different quantum operations $Q_{\xi}$. Evidently, there is a hierarchy of such measurements, partially ordered by inclusion of partitions, and we can speak of coarsenings and refinements.

\subsection{Continuous outcomes}
\label{sec:coutcomes}

In the present section, we consider measurements where we extract the precise value of the observable of interest. Concretely, we want to construct quantum operations that implement such measurements of the observable $D_{\xi}$. Naively, we might expect a quantum operation that selects outcome $q\in\R$ to take the form,
\begin{equation}
    \sigma\mapsto \Pi_{\xi}(q) \sigma \Pi_{\xi}(q) ,
    \label{eq:opqtry}
\end{equation}
with $\Pi_{\xi}(q)$ the operator induced by the sesquilinear form $B_{\xi}(q)$ of equation \eqref{eq:smeasure}, the key ingredient of the spectral decomposition of $\wq{D_{\xi}}$.
Then, weighing with the outcome value, we might expect the quantum operation encoding measurement of the expectation value to take the form,
\begin{equation}
    \sigma\mapsto \int_{-\infty}^{\infty} \xd q\, q\, \Pi_{\xi}(q) \sigma \Pi_{\xi}(q) .
    \label{eq:opqinttry}
\end{equation}
However, the operator $\Pi_{\xi}(q)$ does not exist as the sesquilinear form $B_{\xi}(q)$ does not admit an extension to the whole Hilbert space $\cH$. Thus, expression \eqref{eq:opqtry} is not well-defined and neither is expression \eqref{eq:opqinttry}. On the other hand, we do have a family of POVM decompositions of $\wq{D_{\xi}}$ that approximate the spectral decomposition weakly (Lemmas~\ref{lem:weakcont} and \ref{lem:weakcontint}). This suggests to replacing the would-be operators $\Pi_{\xi}(q)$ by the operators $\Pi_{\xi}^{\epsilon}(q)$ with $\epsilon>0$. This indeed makes expressions~\eqref{eq:opqtry} and \eqref{eq:opqinttry} well-defined (in the latter case depending on $\sigma$). What is more, replacing a PVM with a POVM conserves most features relevant for the interpretation as a quantum measurement. We might envisage taking a limit $\epsilon\to 0$ at the end. It turns out that there is another modification necessary of expression \eqref{eq:opqinttry} and thus also of expression~\eqref{eq:opqtry}: A suitable normalization factor has to be inserted.

\begin{lem}
    \label{lem:defopq}
    For $\epsilon>0$, $\xi\in\cH$, $q\in\R$, $\sigma\in\cT$ define,
    \begin{equation}
        A_{\xi}^{\epsilon}(q)(\sigma)\defeq \sqrt{2\pi}\,\epsilon\, \Pi_{\xi}^{\epsilon}(q) \sigma \Pi_{\xi}^{\epsilon}(q) .
        \label{eq:defopq}
    \end{equation}
    \begin{enumerate}[label=\alph*)]
        \item The map $\cT\to\cT:\sigma\mapsto A_{\xi}^{\epsilon}(q)(\sigma)$ is completely positive, and continuous in the trace norm. Moreover, continuity is uniform in $\xi$ and uniform in $\epsilon$ for $\epsilon\ge c$, where $c>0$ is an arbitrary positive constant.
        \item The map $\R\to\cT:q\mapsto A_{\xi}^{\epsilon}(q)(\sigma)$ is continuous in the trace norm.
        \item The map $\R^+\to\cT:\epsilon\mapsto A_{\xi}^{\epsilon}(q)(\sigma)$ is continuous in the trace norm.
    \end{enumerate}
\end{lem}
\begin{proof}
    For a) complete positivity follows directly from the definition, due to the self-adjointness of $\Pi_{\xi}^{\epsilon}(q)$. As for continuity, due to linearity we consider
    \begin{equation}
        \left\|A_{\xi}^{\epsilon}(q)(\sigma)\right\|_{\ntr}
        =\sqrt{2\pi}\epsilon \left\|\Pi_{\xi}^{\epsilon}(q) \sigma \Pi_{\xi}^{\epsilon}(q)\right\|_{\ntr}
        \le \sqrt{2\pi}\epsilon \|\Pi_{\xi}^{\epsilon}(q)\|_{\nop}^2
        \|\sigma\|_{\ntr}
        \le \sqrt{\frac{2}{\pi}}\frac{1}{\epsilon}\|\sigma\|_{\ntr}
        \le \sqrt{\frac{2}{\pi}}\frac{1}{c}\|\sigma\|_{\ntr} .
    \end{equation}
    Here we have used the bound on $\Pi_{\xi}^{\epsilon}(q)$ of Lemma~\ref{lem:bopq}.
    For b) consider,
    \begin{multline}
        \left\|A_{\xi}^{\epsilon}(q)(\sigma)-A_{\xi}^{\epsilon}(q')(\sigma)\right\|_{\ntr}
        =\sqrt{2\pi}\epsilon \left\|\Pi_{\xi}^{\epsilon}(q) \sigma
        \left(\Pi_{\xi}^{\epsilon}(q)-\Pi_{\xi}^{\epsilon}(q')\right)
        +\left(\Pi_{\xi}^{\epsilon}(q)-\Pi_{\xi}^{\epsilon}(q')\right)
        \sigma \Pi_{\xi}^{\epsilon}(q')\right\|_{\ntr}\\
        \le \sqrt{2\pi}\epsilon \|\sigma\|_{\ntr}
        \left(\|\Pi_{\xi}^{\epsilon}(q)\|_{\nop} + \|\Pi_{\xi}^{\epsilon}(q')\|_{\nop}\right) \|\Pi_{\xi}^{\epsilon}(q)-\Pi_{\xi}^{\epsilon}(q')\|_{\nop}
        \le 2\sqrt{2} \|\sigma\|_{\ntr}
        \|\Pi_{\xi}^{\epsilon}(q)-\Pi_{\xi}^{\epsilon}(q')\|_{\nop} .
    \end{multline}
    We have used the bound on $\Pi_{\xi}^{\epsilon}(q)$ of Lemma~\ref{lem:bopq}. Continuity follows then from Lemma~\ref{lem:normcont}. For c) consider,
    \begin{multline}
        \left\|A_{\xi}^{\epsilon}(q)(\sigma)-A_{\xi}^{\epsilon'}(q)(\sigma)\right\|_{\ntr} \\
        =\sqrt{2\pi} \left\|\sqrt{\epsilon}\Pi_{\xi}^{\epsilon}(q) \sigma
        \left(\sqrt{\epsilon} \Pi_{\xi}^{\epsilon}(q)-\sqrt{\epsilon'} \Pi_{\xi}^{\epsilon'}(q)\right)
        +\left(\sqrt{\epsilon} \Pi_{\xi}^{\epsilon}(q)-\sqrt{\epsilon'} \Pi_{\xi}^{\epsilon'}(q)\right)
        \sigma \sqrt{\epsilon'}\Pi_{\xi}^{\epsilon'}(q)\right\|_{\ntr} \\
        \le \sqrt{2\pi} \|\sigma\|_{\ntr}
        \left(\sqrt{\epsilon}\|\Pi_{\xi}^{\epsilon}(q)\|_{\nop}
         + \sqrt{\epsilon'}\|\Pi_{\xi}^{\epsilon'}(q)\|_{\nop}\right) \|\sqrt{\epsilon} \Pi_{\xi}^{\epsilon}(q)-\sqrt{\epsilon'}\Pi_{\xi}^{\epsilon'}(q)\|_{\nop} \\
        \le \sqrt{2} \|\sigma\|_{\ntr}\left(\frac{1}{\sqrt{\epsilon}}+\frac{1}{\sqrt{\epsilon'}}\right)
        \|\sqrt{\epsilon} \Pi_{\xi}^{\epsilon}(q)-\sqrt{\epsilon'}\Pi_{\xi}^{\epsilon'}(q)\|_{\nop} .
    \end{multline}
    Again, we have used the bound on $\Pi_{\xi}^{\epsilon}(q)$ of Lemma~\ref{lem:bopq}. By Lemma~\ref{lem:normcont} the map $\epsilon\mapsto \Pi_{\xi}^{\epsilon}(q)$ is norm continuous. But this is thus also true for the map $\epsilon\mapsto \sqrt{\epsilon}\Pi_{\xi}^{\epsilon}(q)$. This completes the proof.
\end{proof}

\begin{lem}
    \label{lem:qocomp}
\begin{equation}
    A_{\xi}^{\epsilon'}(q')\circ A_{\xi}^{\epsilon}(q)
    =\sqrt{\frac{2}{\pi (\epsilon^2+\epsilon'^2)}}
    \exp\left(-\frac{2}{\epsilon^2 +\epsilon'^2} (q-q')^2\right)
    A_{\xi}^{\epsilon''}\left(\frac{\epsilon'^2 q + \epsilon q'^2}{\epsilon^2+\epsilon'^2}\right),
    \quad\text{with}\;\; \epsilon''^2=\frac{\epsilon^2 \epsilon'^2}{\epsilon^2+\epsilon'^2} .
    \label{eq:qocomp}
\end{equation}
\end{lem}
\begin{proof}
    This follows from Lemma~\ref{lem:comp} and the definition \eqref{eq:defopq}.
\end{proof}

\begin{lem}
    \label{lem:acommute}
    If $\omega(\xi,\xi')=0$, then $[A_{\xi'}^{\epsilon'}(q'), A_{\xi}^{\epsilon}(q)]=0$.
\end{lem}
\begin{proof}
    This follows from Lemma~\ref{lem:dcommute} and the definition \eqref{eq:defopq}.
\end{proof}

For a Lebesgue measurable function $f:\R\to\C$ and $\sigma\in\cT$ we introduce the following notation, whenever the integral exists in the trace norm topology,
\begin{equation}
    A_{\xi}^{\epsilon}[f](\sigma)
    \defeq \int_{-\infty}^{\infty}\xd q\, f(q)\, A_{\xi}^{\epsilon}(q)(\sigma)
    = \sqrt{2\pi}\,\epsilon \int\xd q\, f(q)\, \Pi_{\xi}^{\epsilon}(q) \sigma \Pi_{\xi}^{\epsilon}(q) .
    \label{eq:defaint}
\end{equation}
Note that $A_{\xi}^{\epsilon}[f]$ is completely positive if $f$ is positive, i.e., if $f\ge 0$.

\begin{lem}
    \label{lem:opbound}
    Let $\sigma\in\cT$ be self-adjoint and $f:\R\to\C$ essentially bounded. Then,
    \begin{equation}
        \left\| A_{\xi}^{\epsilon}[f](\sigma)\right\|_{\ntr}
        \le \|f\|_{\infty} \|\sigma\|_{\ntr} .
        \label{eq:opbound}
    \end{equation}
\end{lem}
\begin{proof}
    Since $\sigma$ is self-adjoint we can decompose it into its positive and negative parts, $\sigma=\sigma^+ - \sigma^-$, where $\sigma^+$ and $\sigma^-$ are positive operators such that $\sigma^+\sigma^-=\sigma^-\sigma^+=0$. We write $|\sigma|=\sigma^+ + \sigma^-$. Then, using Lemmas~\ref{lem:comp} and \ref{lem:fbound},
    \begin{align*}
        \left\| A_{\xi}^{\epsilon}[f](\sigma)\right\|_{\ntr}
        & \le \sqrt{2\pi}\,\epsilon \int_{-\infty}^{\infty}\xd q\, |f(q)|
            \left\| \Pi_{\xi}^{\epsilon}(q) (\sigma^+ - \sigma^-) \Pi_{\xi}^{\epsilon}(q)\right\|_{\ntr} \\
        & \le \sqrt{2\pi}\,\epsilon
            \int_{-\infty}^{\infty}\xd q\, |f(q)| \left(\left\|\Pi_{\xi}^{\epsilon}(q) \sigma^+ \Pi_{\xi}^{\epsilon}(q)\right\|_{\ntr}
            +\left\|\Pi_{\xi}^{\epsilon}(q) \sigma^- \Pi_{\xi}^{\epsilon}(q)\right\|_{\ntr} \right) \\
        & = \sqrt{2\pi}\,\epsilon \int_{-\infty}^{\infty}\xd q\, |f(q)|\,
        \tr\left(\Pi_{\xi}^{\epsilon}(q)\, |\sigma|\, \Pi_{\xi}^{\epsilon}(q)\right) \\
        & = \int_{-\infty}^{\infty}\xd q\, |f(q)|\,
        \tr\left(\Pi_{\xi}^{\epsilon/\sqrt2}(q)\, |\sigma| \right) \\
        & = \tr\left(\Pi_{\xi}^{\epsilon/\sqrt2}[|f|]\, |\sigma| \right)
        = \left\|\Pi_{\xi}^{\epsilon/\sqrt2}[|f|]\, |\sigma| \right\|_{\ntr} \\
        & \le \left\|\Pi_{\xi}^{\epsilon/\sqrt2}[|f|]\right\|_{\Lop} \|\sigma\|_{\ntr}
        \le \|f\|_{\infty} \|\sigma\|_{\ntr} .
    \end{align*}
\end{proof}

\begin{lem}
    \label{lem:aintcom}
    If $\omega(\xi,\xi')=0$, then $[A_{\xi'}^{\epsilon'}[f'], A_{\xi}^{\epsilon}[f]]=0$.
\end{lem}
\begin{proof}
    This follows from Lemma~\ref{lem:acommute} and definition \eqref{eq:defaint}.
\end{proof}

\begin{lem}
    \label{lem:acontepsilon}
    Let $\sigma\in \cT$ and $f:\R\to\C$ essentially bounded. Then, the map $\R^+\to \cT: \epsilon\mapsto A_{\xi}^\epsilon[f](\sigma)$ is continuous.
\end{lem}
\begin{proof}
    Let $\epsilon,\epsilon'>0$,
\begin{align*}
    & \left\|A_{\xi}^{\epsilon}[f](\sigma)-A_{\xi}^{\epsilon'}[f](\sigma)\right\|_{\ntr}
    \le \int_{-\infty}^{\infty}\xd q\, |f(q)| \left\|A_{\xi}^{\epsilon}(q)(\sigma)-A_{\xi}^{\epsilon'}(q)(\sigma)\right\|_{\ntr} \\
    & =\sqrt{2\pi} \int_{-\infty}^{\infty}\xd q\, |f(q)|
       \left\|\sqrt{\epsilon}\Pi_{\xi}^{\epsilon}(q) \sigma
       \left(\sqrt{\epsilon} \Pi_{\xi}^{\epsilon}(q)-\sqrt{\epsilon'} \Pi_{\xi}^{\epsilon'}(q)\right)
       +\left(\sqrt{\epsilon} \Pi_{\xi}^{\epsilon}(q)-\sqrt{\epsilon'} \Pi_{\xi}^{\epsilon'}(q)\right)
       \sigma \sqrt{\epsilon'}\Pi_{\xi}^{\epsilon'}(q)\right\|_{\ntr} \\
    &\le\sqrt{2\pi} \int_{-\infty}^{\infty}\xd q\, |f(q)|\\
    & \quad \left(\left\|\sqrt{\epsilon}\Pi_{\xi}^{\epsilon}(q) \sigma
    \left(\sqrt{\epsilon} \Pi_{\xi}^{\epsilon}(q)-\sqrt{\epsilon'} \Pi_{\xi}^{\epsilon'}(q)\right)\right\|_{\ntr}
    +\left\|\left(\sqrt{\epsilon} \Pi_{\xi}^{\epsilon}(q)-\sqrt{\epsilon'} \Pi_{\xi}^{\epsilon'}(q)\right)
    \sigma \sqrt{\epsilon'}\Pi_{\xi}^{\epsilon'}(q)\right\|_{\ntr}\right) \\
    & \le\sqrt{2\pi}\, \|\sigma\|_{\ntr} \int_{-\infty}^{\infty}\xd q\, |f(q)|
        \left(\sqrt{\epsilon}\left\|\Pi_{\xi}^{\epsilon}(q) \right\|_{\Lop}
        + \sqrt{\epsilon'}\left\|\Pi_{\xi}^{\epsilon'}(q)\right\|_{\Lop}\right) \left\|\sqrt{\epsilon} \Pi_{\xi}^{\epsilon}(q)-\sqrt{\epsilon'}  \Pi_{\xi}^{\epsilon'}(q)\right\|_{\Lop} \\
    & \le\sqrt{2\pi}\, \|\sigma\|_{\ntr}\,  c(\epsilon,\epsilon') \int_{-\infty}^{\infty}\xd q\, |f(q)|
        \left(\sqrt{\epsilon}\left\|\Pi_{\xi}^{\epsilon}(q)\right\|_{\Lop}
        +\sqrt{\epsilon'}\left\|\Pi_{\xi}^{\epsilon'}(q)\right\|_{\Lop}\right) \\
    & \le\sqrt{2\pi}\, \|\sigma\|_{\ntr}\,  c(\epsilon,\epsilon')
        \left(\sqrt{\epsilon}\left\|\Pi_{\xi}^{\epsilon}[|f|]\right\|_{\Lop}
        +\sqrt{\epsilon'}\left\|\Pi_{\xi}^{\epsilon'}[|f|]\right\|_{\Lop}\right) \\
    & \le\sqrt{2\pi}\, \|\sigma\|_{\ntr}\,  c(\epsilon,\epsilon')
        \left(\sqrt{\epsilon}
        +\sqrt{\epsilon'}\right) \|f\|_{\infty} .
\end{align*}
We use an estimate similar to \eqref{eq:piqepsilon}. Namely,
\begin{equation}
    \left\|\sqrt{\epsilon}\Pi_{\xi}^{\epsilon}(q)
    -\sqrt{\epsilon'}\Pi_{\xi}^{\epsilon'}(q)\right\|_{\nop}
    \le c(\epsilon,\epsilon')\defeq
    \frac{1}{\pi}\int_{-\infty}^{\infty}\xd t\, \left| \sqrt{\epsilon'} e^{-\epsilon'^2 t^2}-\sqrt{\epsilon}e^{-\epsilon^2 t^2}\right| .
\end{equation}
Crucially, the bound $c(\epsilon,\epsilon')$ converges to $0$ when $\epsilon$ and $\epsilon'$ approach each other.
\end{proof}

\begin{lem}
    \label{lem:tracereduce}
    \begin{equation}
        \tr\left(A_{\xi}^{\epsilon}[f](\sigma)\right)
        =\tr\left(\Pi_{\xi}^{\epsilon/\sqrt{2}}[f]\, \sigma\right) .
    \end{equation}
\end{lem}
\begin{proof}
    This follows by explicit calculation with Lemma~\ref{lem:comp}.
\end{proof}

\begin{lem}
    \label{lem:tralim}
    \begin{equation}
        \lim_{\epsilon\to 0}\tr\left(A_{\xi}^{\epsilon}[f](\sigma)\right)
        =\tr\left(\Pi_{\xi}[f]\, \sigma\right) .
    \end{equation}
\end{lem}
\begin{proof}
    With Lemma~\ref{lem:tracereduce} it remains to show
    \begin{equation}
        \lim_{\epsilon\to 0}\tr\left(\Pi_{\xi}^{\epsilon/\sqrt{2}}[f]\, \sigma\right)=\tr\left(\Pi_{\xi}[f]\, \sigma\right) .
    \end{equation}
    First, take $\sigma$ to be a projection operator, then this follows from the weak convergence result of Lemma~\ref{lem:weakcontint}. Now recall $|\tr(A\sigma)|\le \| A\|_\nop \|\sigma\|_{\tr}$ for a bounded operator $A$. Since Lemma~\ref{lem:fbound} gives an estimate of $\|\Pi_{\xi}^{\epsilon}[f]\|_{\nop}$ independent of $\epsilon$, we can approximate $\sigma$ by linear combinations of projectors in the trace norm.
\end{proof}

\begin{lem}
    \label{lem:traceop}
    Let $\epsilon>0$, $\sigma\in\cT$.
    \begin{align}
        \tr\left(A_{\xi}^{\epsilon}[\one](\sigma)\right)
          & =\tr(\sigma), \label{eq:trconst} \\
        \tr\left(A_{\xi}^{\epsilon}[\id](|\coh_{\beta}\rangle\langle\coh_{\gamma}|)\right)
          & =\langle\coh_\gamma, \coh_\beta\rangle D_{\xi}(\beta,\gamma)
            =\langle\coh_\gamma, \wq{D_{\xi}}\coh_\beta\rangle
            =\tr\left(\wq{D_{\xi}} |\coh_{\beta}\rangle\langle\coh_{\gamma}|\right),
            \label{eq:trlin} \\
        \tr\left(A_{\xi}^{\epsilon}[e_s](|\coh_\beta\rangle\langle\coh_\gamma|)\right)
          & = \langle\coh_\gamma, \coh_\beta\rangle \exp\left(\im s D_{\xi}(\beta,\gamma)\right) \exp\left(-\frac{s^2}{4} \left(\|\xi\|^2+\frac{\epsilon^2}{2}\right)
         \right) .
         \label{eq:trchar}
    \end{align}
\end{lem}
\begin{proof}
    This follows combining Lemmas~\ref{lem:tracereduce} and \ref{lem:integrals}.
\end{proof}

We now have at our disposal quantum operations that encode a precise measurement of the value of the observable $D_{\xi}$. For fixed $\epsilon>0$ the quantum operation that selects outcome $q\in\R$ is $A_{\xi}^{\epsilon}(q)$ as defined in expression~\eqref{eq:defopq}.
Since a single point $q$ has measure zero in $\R$ it generally does not make sense to use the quantum operation $A_{\xi}^{\epsilon}(q)$ directly. Instead, the relevant quantum operation is the weighted version $A_{\xi}^{\epsilon}[f]$ with weight function $f$, see expression~\eqref{eq:defaint}.
In particular, the \emph{non-selective quantum operation} that corresponds to the measurement without reading any outcome is $A_{\xi}^{\epsilon}[\one]$. It is the requirement that $A_{\xi}^{\epsilon}[\one]$ be trace-preserving that fixes the numerical factor in definition~\eqref{eq:defaint}. This has been carried over also to definition~\eqref{eq:defopq}.
Let $X\subseteq\R$ be Lebesgue measurable. Then, the \emph{selective quantum operation} that tests whether the outcome lies in the subset $X$ is $A_{\xi}^{\epsilon}[\chi_X]$. In particular, $X\mapsto A_{\xi}^{\epsilon}[\chi_X]$ defines a \emph{quantum instrument} \cite{DaLe:opapquantprob}. On the other hand, the quantum operation encoding measurement of the expectation value is $A_{\xi}^{\epsilon}[\id]$.

Crucially, for each $\epsilon>0$, all quantum operations refer to the same measurement, which extracts the precise value of the observable $D_{\xi}$ from the quantum system. This is in contrast to the discrete outcome setting of Section~\ref{sec:doutcomes}, where each partition defines a different measurement with a different amount of information extracted from the quantum system. In particular, the quantum operations $Q_{\xi}[X]$ and $A_{\xi}^{\epsilon}[\chi_X]$ are fundamentally different, even though they apparently answer the same question. In the first case, only information whether the value lies in $X$ is extracted from the quantum system. In the second case, the exact value is extracted from the quantum system, and then it is checked if this value lies in $X$. For the same reason, there is no analog in the discrete outcome setting of the weighted quantum operation $A_{\xi}^{\epsilon}[\id]$ for the expectation value, although it would be possible to approximate it.

Instead of using the spectral decomposition of $\wq{D_{\xi}}$ to construct quantum operations for measuring it, we have used a 1-parameter family of POVM decompositions that approximate the spectral decomposition. This raises several questions. On the one hand, one may ask how good the constructed quantum operations are in measuring the observable $D_{\xi}$. A partial answer to this is given by Lemma~\ref{lem:traceop}, relation~\eqref{eq:trlin}. If we start with an initial state, measure, and then discard, the quantum operation $A_{\xi}^{\epsilon}[\id]$ corresponds exactly to measuring the expectation value of the observable $D_{\xi}$, independently of the parameter $\epsilon$. That is for a single measurement, we obtain a prefect result for all values of $\epsilon$. Of course, once we consider composites of various measurements, a dependence on the parameter(s) $\epsilon$ will generically appear.

Another question is, how the quantum operations behave under change of the parameter $\epsilon$, and whether we can take a limit $\epsilon\to 0$. As to the first question, Lemma~\ref{lem:acontepsilon} provides a continuity result under a change of $\epsilon$. As to the second question, while it seems a limit $\lim_{\epsilon\to 0} A_{\xi}^{\epsilon}[f]$ as a quantum operation does not in general exist, this is not really necessary. What we need instead is that the limit of those expressions that yield probabilities or expectation values of measurements do exist. This is a much weaker requirement, weaker even than the existence of the limit $\lim_{\epsilon\to 0} A_{\xi}^{\epsilon}[f](\sigma)$ for a fixed state $\sigma\in\cT$. In this sense Lemma~\ref{lem:tralim} is promising as it shows the existence of a limit under the trace.

Finally, we consider the question whether the measurements defined here are \emph{ideal}. While clearly the quantum operations for $\epsilon>0$ are not ideal, as can be read off from Lemma~\ref{lem:qocomp}, the same result suggests that the limit $\epsilon\to 0$ is ideal. Observe first that for $\epsilon'=\epsilon$ expression \eqref{eq:qocomp} simplifies to,
\begin{equation}
    A_{\xi}^{\epsilon}(q')\circ A_{\xi}^{\epsilon}(q)
    =\frac{1}{\sqrt{\pi} \epsilon}
    \exp\left(-\frac{1}{\epsilon^2} (q-q')^2\right)
    A_{\xi}^{\epsilon/\sqrt2}\left(\frac{q + q'}{2}\right) .
\end{equation}
Informally, we can take the limit of this expression to obtain
\begin{equation}
    "\lim_{\epsilon\to 0} A_{\xi}^{\epsilon}(q')\circ A_{\xi}^{\epsilon}(q)
    = \delta\left(\frac12(q-q)'\right) \lim_{\epsilon\to 0} A_{\xi}^{\epsilon}\left(q\right)" .
\end{equation}
In particular this suggests,
\begin{equation}
    "\lim_{\epsilon\to 0} A_{\xi}^{\epsilon}[\chi_A]\circ A_{\xi}^{\epsilon}[\chi_A]
    = \lim_{\epsilon\to 0} A_{\xi}^{\epsilon}[\chi_A]" .
\end{equation}
We leave rigorous considerations for future work.


\section{Locality and Causality}
\label{sec:loccaus}

We suppose that spacetime is a globally hyperbolic manifold and that we have a foliation of spacetime by spacelike Cauchy hypersurfaces labeled by a global time function. We denote the space of germs of solutions of the equations of motion in a neighborhood of the hypersurface $\Sigma_t$ at time $t$ by $L_t$. Equivalently, $L_t$ is the space of initial data at time $t$. The Fock space $\cH_t$ over $L_t$ is the space of states at time $t$. Crucially, the symplectic form $\omega_t$ on $L_t$ is a local expression. That is, if $\phi,\phi'\in L_t$ have support on disjoint subsets in $\Sigma_t$, then $\omega_t(\phi,\phi')=0$. What is more, the correspondence \eqref{eq:linobsdual} between linear observables on $L_t$ and elements of $L_t$ preserves this notion of locality. That is, the linear observable $D_\xi$ vanishes on germs that have a support disjoint from the support of $\xi$.
In the quantum theory this notion of locality is reflected in the commutators of field operators. We have from \eqref{eq:linca} and \eqref{eq:ccr},
\begin{equation}
    [\wq{D_{\xi'}},\wq{D_{\xi}}]=2\im\omega(\xi,\xi') .
    \label{eq:comhs}
\end{equation}
In particular, if $\xi$ and $\xi'$ have disjoint support on the hypersurface, the commutator vanishes.

In the following, it is convenient to move to a picture where, as usual, we make reference only to a single Hilbert space $\cH$ of states, e.g., by arbitrarily selecting one spacelike hypersurface. In the classical theory we use the correspondence between germs on a spacelike hypersurface and global solutions to identify observables on spaces of germs with observables on global solutions. The conservation of the symplectic form in time gives rise to a symplectic form on global solutions. The correspondence \eqref{eq:linobsdual} thus extends to a correspondence between linear observables on the space of global solutions (that we also call) $L$ and elements of $L$. Crucially, this remains true in the quantum theory. That is, let $\xi\in L$ be a global solution that restricts to germs $\xi_1\in L_1$ at time $t_1$ and $\xi_2\in L_2$ at time $t_2$, where $t_2> t_1$. Then, the operator $\wq{D_{\xi_2}}$ on $\cH_2$ is precisely the time-evolved version of the operator $\wq{D_{\xi_1}}$ on $\cH_1$, $\wq{D_{\xi_2}}=U_{[t_1,t_2]}\wq{D_{\xi_1}}U_{[t_1,t_2]}^{\dagger}$. We may thus identify these operators, giving rise to a well-defined operator $\wq{D_{\xi}}$ on $\cH$, where $\xi\in L$ is a global solution. While the commutator equation \eqref{eq:comhs} acquires in this way a global interpretation, its hypersurface specific interpretation is still valid, for any spacelike hypersurface. In particular, this means that given two global solutions $\xi,\xi'\in L$, if there is any spacelike hypersurface on which these solutions have disjoint supports, the commutator $[\wq{D_{\xi'}},\wq{D_{\xi}}]$ vanishes.

\subsection{Locality}
\label{sec:locality}

We say that a global solution $\xi\in L$ is \emph{localizable} in a subset $S$ of a spacelike hypersurface $\Sigma$ if its restriction to a germ on the hypersurface has support inside the subset. We assume locality and non-degeneracy of the symplectic form, so that $\xi$ is localizable in $S$ if and only if the symplectic from with any solution $\xi'$ that is localizable in the complement $\Sigma\setminus S$ vanishes, $\omega(\xi,\xi')=0$. Similarly, we say that a linear observable is \emph{localizable} in $S$ if and only if the corresponding solution $\xi\in L$ is localizable in $S$. Correspondingly, in the quantum theory we say that a field operator $\wq{D_{\xi}}$ is \emph{localizable} in $S$ if and only if the solution $\xi$ is localizable in $S$. Then, a field operator $E$ is localizable in $S$ if and only if it commutes with any field operator localizable in $\Sigma\setminus S$, due to relation \eqref{eq:comhs}. We extend this characterization to any operator $E$ on the Hilbert space $\cH$. That is, we say that $E$ is \emph{localizable} in $S$ if and only if $E$ commutes with any field operator localizable in $\Sigma\setminus S$.\footnote{It would be mathematically cleaner to use bounded operators for this definition, such as the unitaries $\exp(\im \wq{D_{\xi}})$. However, we feel this would unnecessarily complicate the discussion and leave these refinements to the knowledgeable reader.}

Remarkably, the results of Section~\ref{sec:povmdec} show that not only the operators arising in the spectral decomposition of the filed operator $\wq{D_\xi}$, but also those arising in the $\epsilon$-family of POVM decompositions have precisely the same locality properties as $\wq{D_\xi}$ itself, or equivalently as the classical observable $D_{\xi}$ or the solution $\xi$. This is particularly due to Lemmas~\ref{lem:prod}, \ref{lem:integrals}, and \ref{lem:dcommute} as well as the definition~\eqref{eq:defpint}. Thus, for $\epsilon>0$ the operator $\Pi_{\xi}^{\epsilon}(q)$ is localizable wherever $\xi$ is localizable. Similarly, $\Pi_{\xi}^{\epsilon}[f]$ is localizable wherever $\xi$ is localizable (for any $\epsilon$).

It is straightforward to extend this notion of locality to quantum operations. Thus, we say that a quantum operation $R:\cT\to\cT$ is \emph{localizable} in a subset $S$ of a spacelike hypersurface $\Sigma$ if and only if it commutes both ways with any field operator $\wq{D_{\xi}}$ localizable in $\Sigma\setminus S$, i.e., $R(\wq{D_{\xi}}\sigma)=\wq{D_{\xi}} R(\sigma)$ and $R(\sigma\wq{D_{\xi}})= R(\sigma) \wq{D_{\xi}}$. It is then easy to see from the definitions~\eqref{eq:projop}, \eqref{eq:defopq} and \eqref{eq:defaint} of Section~\ref{sec:quantop} that all quantum operations defined in connection with the observable $D_{\xi}$ have the same locality properties as $\wq{D_{\xi}}$, $D_{\xi}$ and $\xi$ itself, see also Lemma~\ref{lem:aintcom}.

\begin{prop}
  Let $\xi\in L$ be localizable in a subset $S$ of a spacelike hypersurface $\Sigma$. Then, $\Pi_{\xi}[f]$ and $Q_{\xi}[X]$ are localizable in $S$. Also, for $\epsilon>0$, $\Pi_{\xi}^{\epsilon}(q)$, $\Pi_{\xi}^{\epsilon}[f]$, $A_{\xi}^{\epsilon}(q)$ and $A_{\xi}^{\epsilon}[f]$ are localizable in $S$.
\end{prop}

\begin{figure}
  \centering
  \includegraphics[width=0.6\textwidth]{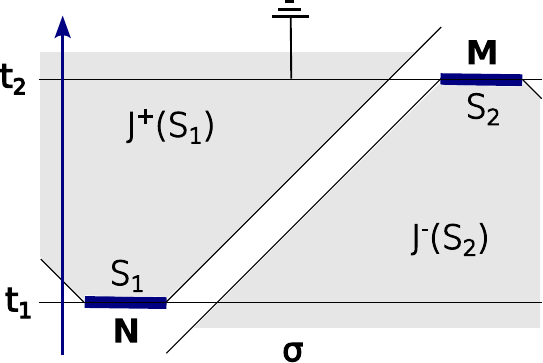}
\caption{Setup with two measurements $N$ and $M$. $N$ is non-selective, while $M$ is selective. By locality, the outcome of $M$ should not depend on whether or not $N$ is performed.}
\label{fig:locality}
\end{figure}

We proceed to illustrate this notion of locality in a measurement context. Suppose we do two measurements, see Figure~\ref{fig:locality}. At time $t_1$ we do a non-selective measurement $N$ that is localizable in a subset $S_1$ of the equal-time hypersurface at $t_1$. At a later time $t_2$ we do a selective measurement $M$ that is localizable in a subset $S_2$ of the equal-time hypersurface at $t_2$. We further assume that the causal future $J^+(S_1)$ of $S_1$ does not intersect $S_2$. This means that $N$ is localizable at time $t_2$ in a set disjoint from $S_2$. Equivalently, we assume that the causal past $J^-(S_2)$ does not intersect $S_1$. This means that $M$ is localizable at time $t_1$ in a set disjoint from $S_1$. With our previous definitions, this implies that the quantum operations $N$ and $M$ commute. Suppose we have an initial state $\sigma$ at $t_1$. Then, the outcome of the measurement is,
\begin{equation}
    \tr(M\circ N(\sigma))=\tr(N\circ M(\sigma))=\tr(M(\sigma)) .
\end{equation}
In particular, the measurement outcome is the same as if the measurement $N$ was not performed.
That is, a non-selective measurement $N$ localizable outside the causal past of a set where a selective measurement $M$ can be localized, does not influence its results.

In the standard formulation of quantum theory the \emph{causality axiom} can be interpreted as the statement that a non-selective measurement cannot affect the outcome probabilities of measurements performed before. This is mathematically implemented by the requirement that a non-selective quantum operation be trace preserving. The present notion of \emph{locality} can be interpreted as the natural relativistic extension of this principle. Thus, a non-selective measurement localizable on a subset of a spacelike hypersurface cannot affect the outcome probabilities of measurements performed outside the causal future of the subset.

\subsection{Causal transparency}
\label{sec:caustransp}

We turn to another constraint that a physically realizable notion of measurement must satisfy: The dynamics of the measurement must not only respect the causal structure of spacetime outside the set where the measurement is localized (this is what locality amounts to), but also inside. In 1993, Rafael Sorkin demonstrated in a seminal work that a projective measurement quite generically violates this condition by enabling superluminal signaling \cite{Sor:impossible}. This is detected through the transmission of a signal between two other measurements that cannot normally communicate causally.

\begin{figure}
  \centering
  \includegraphics[width=0.6\textwidth]{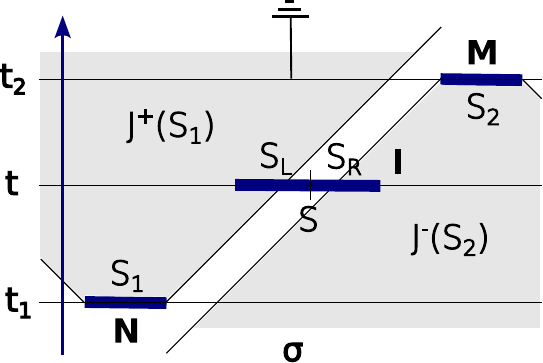}
\caption{Setup with three measurements $N$, $I$ and $M$. $N$ and $I$ are non-selective, while $M$ is selective. Given locality, causal transparency of $I$ means that the outcome of $M$ should not depend on whether or not $N$ is performed.}
\label{fig:caustrans}
\end{figure}

The setup is a modification of the previously discussed setup of Figure~\ref{fig:locality}. At an intermediate time between the initial measurement $N$ and the final measurement $M$, another non-selective measurement $I$ is inserted, see Figure~\ref{fig:caustrans}. This is localizable in a set $S$, which intersects both the causal future of $S_1$, where $N$ is localizable, and the causal past of $S_2$, where $M$ is localizable. Given that $I$ and $M$ are performed, the question is whether we can tell from the outcome at $M$, whether or not the measurement $N$ was performed. If this is the case, the implication is that $I$ has enabled superluminal signaling, which we can think of as happening in $S$. Sorkin has shown for a large class of projective measurements in scalar field theory (with $S$ covering the whole hypersurface at the intermediate time) that such superluminal signaling takes place. This result has long been taken as an indication that conceiving of a measurement theory for quantum field theory as analogous to that of non-relativistic quantum mechanics is problematic. In particular, it suggests that the route of constructing quantum operations through spectral decompositions of self-adjoint operators might not be viable, as these would constitute essentially projective measurements.

We will call measurements that do not lead to superluminal signaling in any scenario of the type described \emph{causally transparent}, as they are "transparent" to the casual structure of spacetime.\footnote{Sorkin used the term "locality". However, we need a different term to differentiate this property from the notion of locality considered in Section~\ref{sec:locality}.} We show in the following that, contrary to expectations, the non-selective quantum operations $A_{\xi}^{\epsilon}[\one]$ encoding the measurement of the linear observable $D_{\xi}$ are causally transparent.

\begin{thm}
  \label{thm:caustransp}
  Let $S_1,S,S_2$ be subsets of the equal-time hypersurfaces at distinct times $t_1,t, t_2$ such that $S_2$ does not intersect the causal future of $S_1$. Let $N$ be a non-selective quantum operation localizable at $S_1$, $M$ a selective quantum operation localizable at $S_2$ and $\xi\in L_t$ localizable in $S$. Let $\epsilon>0$ and $I=A_{\xi}^{\epsilon}[\one]$ and $\sigma\in\cT$. Then,
  \begin{equation}
     \tr\left((M\comp I\comp N)(\sigma)\right)=\tr\left((M\comp I)(\sigma)\right) .
     \label{eq:ctid}
  \end{equation}
  Here, the symbol $\diamond$ means that composition according to the temporal order of operations is applied.
\end{thm}
\begin{proof}
  If we do not have $t_1< t< t_2$, the statement follows already from locality. Thus, assume $t_1< t< t_2$. Then, there is a partition of the set $S$ into a disjoint union $S=S_L\cup S_R$ such that $S_L$ does not intersect the causal past of $S_2$ and $S_R$ does not intersect the causal future of $S_1$. Correspondingly, there is a decomposition $\xi=\xi_L + \xi_R$, such that $\xi_L$ is localizable in $S_L$ and $\xi_R$ is localizable in $S_R$. Set $\epsilon'=\epsilon/\sqrt2$. Using the commutation properties due to localizability, Lemma~\ref{lem:projsumid}, the identity $\Pi_{-\xi}^{\epsilon}(q)=\Pi_{\xi}^{\epsilon}(-q)$ and the fact that $\Pi_0^\epsilon(q)$ is a multiple of the identity, we find,
  \begin{align*}
    \tr\left((M\comp I\comp N)(\sigma)\right)
     & =\tr\left((M\circ I\circ N)(\sigma)\right) \\
     & =\tr\left(\left(M\circ A_{\xi}^{\epsilon}[\one]\circ N\right)(\sigma)\right) \\
     & =\sqrt{2\pi}\epsilon\int_{-\infty}^{\infty}\xd q\,
       \tr\left(M\left(\Pi_{\xi}^{\epsilon}(q) N (\sigma) \Pi_{\xi}^{\epsilon}(q)\right)\right) \\
     & =\sqrt{2\pi}\epsilon\int_{-\infty}^{\infty}\xd q\,\xd s\,\xd s'\,
       \tr\left(M\left(\Pi_{\xi_L}^{\epsilon'}(q-s)\Pi_{\xi_R}^{\epsilon'}(s)
        N(\sigma) \Pi_{\xi_R}^{\epsilon'}(s')\Pi_{\xi_L}^{\epsilon'}(q-s')\right)\right) \\
     & =\sqrt{2\pi}\epsilon\int_{-\infty}^{\infty}\xd q\,\xd s\,\xd s'\,
       \tr\left(\Pi_{\xi_L}^{\epsilon'}(q-s')\Pi_{\xi_L}^{\epsilon'}(q-s)\, M\circ N\left(\Pi_{\xi_R}^{\epsilon'}(s)
        \sigma \Pi_{\xi_R}^{\epsilon'}(s')\right)\right) \\
     & =\sqrt{2\pi}\epsilon\int_{-\infty}^{\infty}\xd q\,\xd s\,\xd s'\,
       \tr\left(\Pi_{-\xi_L}^{\epsilon'}(s'-s-q)\Pi_{\xi_L}^{\epsilon'}(q)\, N\circ M\left(\Pi_{\xi_R}^{\epsilon'}(s)
        \sigma \Pi_{\xi_R}^{\epsilon'}(s')\right)\right) \\
     & =\sqrt{2\pi}\epsilon\int_{-\infty}^{\infty}\xd s\,\xd s'\,
       \tr\left(\Pi_{0}^{\epsilon}(s'-s)\, N\circ M\left(\Pi_{\xi_R}^{\epsilon'}(s)
        \sigma \Pi_{\xi_R}^{\epsilon'}(s')\right)\right) \\
     & =\sqrt{2\pi}\epsilon\int_{-\infty}^{\infty}\xd s\,\xd s'\,
       \tr\left(N\left(\Pi_{0}^{\epsilon}(s'-s)\, M\left(\Pi_{\xi_R}^{\epsilon'}(s)
        \sigma \Pi_{\xi_R}^{\epsilon'}(s')\right)\right)\right) \\
     & =\sqrt{2\pi}\epsilon\int_{-\infty}^{\infty}\xd s\,\xd s'\,
       \tr\left(\Pi_{0}^{\epsilon}(s'-s)\, M\left(\Pi_{\xi_R}^{\epsilon'}(s)
        \sigma \Pi_{\xi_R}^{\epsilon'}(s')\right)\right) \\
     & =\sqrt{2\pi}\epsilon\int_{-\infty}^{\infty}\xd q\,\xd s\,\xd s'\,
       \tr\left(\Pi_{\xi_L}^{\epsilon'}(q-s')\Pi_{\xi_L}^{\epsilon'}(q-s)\, M\left(\Pi_{\xi_R}^{\epsilon'}(s)
        \sigma \Pi_{\xi_R}^{\epsilon'}(s')\right)\right) \\
     & =\sqrt{2\pi}\epsilon\int_{-\infty}^{\infty}\xd q\,\xd s\,\xd s'\,
       \tr\left(M\left(\Pi_{\xi_L}^{\epsilon'}(q-s)\Pi_{\xi_R}^{\epsilon'}(s)
        \sigma \Pi_{\xi_R}^{\epsilon'}(s')\Pi_{\xi_L}^{\epsilon'}(q-s')\right)\right) \\
     & =\tr\left((M\circ I)(\sigma)\right)=\tr\left((M\comp I)(\sigma)\right) .
  \end{align*}
\end{proof}

Crucially, the identity \eqref{eq:ctid} holds for any $\epsilon>0$. Thus, in particular, it applies to the limit $\epsilon\to 0$, if it exists.


\section{Discussion and Outlook}
\label{sec:outlook}

We construct in the first part of this work (Sections~\ref{sec:motivation} and \ref{sec:povmdec}) the spectral decomposition of field operators of free bosonic quantum field theory in an explicit form (Theorem~\ref{thm:povmspecdec}). What is more, we construct a norm-continuous (Lemmas~\ref{lem:normcont} and \ref{lem:strongcontint}) one-parameter family of positive-operator-valued measure (POVM) decompositions (Theorem~\ref{thm:povmspecdec}), having the spectral decomposition as a weak limit (Lemmas~\ref{lem:weakcont} and \ref{lem:weakcontint}). This family also exhibits a semigroup structure (Lemma~\ref{lem:comp}). Crucially, in contrast to the spectral measure, the POVM measures yield well-defined bounded operators at points (Lemma~\ref{lem:bopq}).

In the second part of this work (Section~\ref{sec:quantop}) we construct quantum operations encoding measurements of the linear observable corresponding to the field operator. We distinguish two types: One in which the real line of possible values is partitioned into subsets representing discrete outcomes. The quantum operations distinguishing between these discrete outcomes are readily obtained from the projection operators arising from integrating the spectral measure over the corresponding subsets (Section~\ref{sec:doutcomes}). The other type corresponds to a fully fine-grained measurement of the real value of the observable with continuous outcome (Section~\ref{sec:coutcomes}). Due to the singular nature of the spectral measure at points, this cannot be constructed directly from the spectral measure. Rather, we use the one-parameter family of POVM measures, obtaining in this way a trace-norm-continuous one-parameter family of generators of quantum operations (Lemmas~\ref{lem:defopq} and \ref{lem:acontepsilon}). While the "true" measurement of the observable should correspond to the limit where the POVM measures approximate the spectral measure, this limit does not exist for the quantum operations in the trace norm. However, a much weaker notion of limit is required to describe measurement processes. Only for the composite expression of a probability or expectation value does this limit need to exist. Lemma~\ref{lem:tralim} implies the existence of the limit under the trace and thus provides strong evidence for its existence in the cases of interest. What is more, we have provided evidence that the limit defines an ideal measurement. A detailed examination of these questions is out of scope for the present paper, and should be addressed in future investigations.

Locality and causality properties of the constructed quantum operations are examined in the third part of this work (Section~\ref{sec:loccaus}). We show that not only the spectral measure, but also the family of POVMs have the same locality properties as the field operator to which they correspond. That is, the measures and operators constructed from them are localizable in the same subsets of a spacelike hypersurface as the field operator, and have the same commutation properties. What is more, all quantum operations, both of the discrete and the continuous outcome type, inherit these locality properties. In particular, a non-selective measurement performed outside the causal past of a selective one, cannot influence the outcome statistics of the latter, assuming a subsequent discard (Section~\ref{sec:locality}).

A causality property crucial for measurements to be physically implementable is that the dynamics of the measurement must not violate special relativity. This would be the case if the measurement enables signaling (which would thus be superluminal) between other measurements in relatively spacelike separated spacetime regions. We call this principle \emph{causal transparency} (Section~\ref{sec:caustransp}). We provide a formal proof (Theorem~\ref{thm:caustransp}) that the one-parameter family of non-selective quantum operations of the continuous outcome type constructed in Section~\ref{sec:coutcomes} satisfies causal transparency. This implies in particular, that the limit where this family approaches the "true" measurement of the corresponding observable, if it exists (see above), also satisfies causal transparency.

To put the latter results into perspective, we specialize to Klein-Gordon theory which is often used to discuss questions of locality and causality in the literature. Thus, any field operator can be written in the form 
\begin{equation}
  \hat{\phi}(f,g)=\int\xd^3 x\, \left(f(x)\hat{\phi}(t,x) + g(x)\hat{\dot{\phi}}(t,x)\right)  .
\end{equation}
The support of the field operator is the union of the supports of the functions $f$ and $g$ on the hypersurface at $t$ in terms of our discussion of locality in Section~\ref{sec:locality}. In particular, the elements of the phase space $L$ are in correspondence to pairs $(f,g)$ and allow for a rather straightforward interpretation in terms of initial data. On the other hand, it is often more convenient to use a smearing function $f$ in spacetime rather than space, which also allows to omit the operator $\hat{\dot{\phi}}(t,x)$. In this setting, variations of our continuous outcome measurement family have appeared previously in the literature in the context of weak measurements, see \cite{Jub:causalupdates,Bed:generalmeasure} and references therein. What is more, Jubb has already argued that such measurements satisfy causal transparency \cite{Jub:causalupdates}.

Sorkin showed in a seminal work that a large class of projection valued measurements violate causal transparency \cite{Sor:impossible}. Since Sorkin's argument was quite generic, this has been taken to mean that any projection-valued measurement probably violates causal transparency. Since quantum operations corresponding to self-adjoint operators are normally constructed out of projectors arising from the spectral decomposition of the operator, this has been interpreted as an important obstacle in obtaining a reasonable theory of quantum measurement for quantum field theory \cite{PaFr:eliminatingimpossible}.
In the context of the present work, this affects in particular the discrete outcome measurements discussed in Section~\ref{sec:doutcomes}. Indeed, the conclusion that they violate causal transparency has been reinforced in a recent more specific analysis of these types of measurements in Klein-Gordon theory by Albertini and Jubb \cite{AlJu:measurecausal}. This makes our proof of causal transparency for the continuous outcome measurements surprising, as one might think of these as arising from an (infinite) refinement of the discrete outcome measurements.
In particular, our result suggests that a satisfactory theory of measurement, broadly analogous to that of non-relativistic quantum mechanics, but respecting locality and causality, can be constructed for quantum field theory after all.

One of the first questions that arises in constructing such a theory is whether and how we can measure observables other than linear ones. The functional calculus of operations that we have introduced in Section~\ref{sec:coutcomes} suggests one avenue to answer this question. It tells us in particular how we can measure any observable that arises as a function of a linear one. What is more, the corresponding non-selective measurement is the same one independent of the function and in particular satisfies causal transparency. Moreover, this approach can be extended to observables that admit representations as functions of several linear observables. Thus, the quantum operation \eqref{eq:defaint} generalizes straightforwardly to the following one:
\begin{equation}
  A_{\xi_1,\ldots,\xi_n}^{\epsilon_1,\ldots\epsilon_n}[f]
  \defeq \int_{\R^n}\xd q_1\cdots\xd q_n\, f(q_1,\ldots,q_n)\, A_{\xi_n}^{\epsilon_n}(q_n)\circ\cdots\circ A_{\xi_1}^{\epsilon_1}(q_1) .
  \label{eq:opcomp}
\end{equation}
Crucially, the corresponding non-selective quantum operation is simply the composition of the non-selective quantum operations for $\xi_1,\ldots,\xi_n$, i.e., $A_{\xi_n}^{\epsilon_n}[\one]\circ\cdots\circ A_{\xi_1}^{\epsilon_1}[\one]$ and thus causally transparent.

We have already seen (in Section~\ref{sec:coutcomes}) that an important difference to the usual non-relativistic recipes of measurement lies in the fact that the full value of the linear observable is extracted from the quantum system, even if we are interested in a function of the observable that forgets some of this information. For example, if we want to measure the square $D^2$ of a linear observable $D$, we also extract its sign, whereas in the usual scheme based on the square $\widehat{D}^2$ of the operator $\widehat{D}$, this is not the case. What difference does this make to expectation values? For a single measurement, that is, if we just prepare, then measure, then discard, this is easy to answer. As can be read off from comparing expressions~\eqref{eq:charspec} and \eqref{eq:trchar} for any function of the observable there is no difference at all in the expectation value in the limit $\epsilon\to 0$. For joint probabilities and expectation values of composite measurements there will certainly be differences. However, there we have much less evidence what the "right" predictions should be. Indeed, physical realizability of the measurement (in the form of casual transparency) might be a more important principle in choosing a scheme than sticking to recipes from the non-relativistic realm. In any case, the development and practical application of the proposed approach will show its merits, or not.

An important direction for generalizing the presented approach is to consider observables that are extended not only in space, but also in time. To this end, it is convenient to switch to a fully covariant formalism, and a natural choice is the path integral. An observable in a spacetime region $M$ is then given by a map $K_M\to\R$, where $K_M$ denotes the space of field configurations in $M$. Note that the previously used notion of observable on the phase space $L$ on a (usually spacelike) hypersurface $\Sigma$ can be recovered as a degenerate case, where the region $M$ is squeezed to an infinitesimal neighborhood of $\Sigma$ \cite{Oe:feynobs,CoOe:locgenvac}. (This is called a \emph{slice observable}.) Since the path integral implements Weyl quantization (for the case of slice observables), the spacetime analog of the operator $\Pi_{\xi}^{\epsilon}(q)$ is obtained by inserting the spacetime analog of the observable $H_{\xi}^{\epsilon}(q)$ of Proposition~\ref{prop:piobs} into the path integral. For a linear observable $D:K_M\to\R$ this observable thus takes the form
\begin{equation}
    H_{D}^{\epsilon}(q)(\phi)
    \defeq\frac{1}{\sqrt{\pi}\epsilon}\exp\left(-\frac{1}{\epsilon^2}(D(\phi)-q)^2\right) .
\end{equation}
With this we can define the spacetime analogs of the quantum operations \eqref{eq:defopq} and \eqref{eq:defaint}. Spacetime analogs of quantum operations are called \emph{probes} in the \emph{positive formalism}, which is the natural spacetime generalization of the non-relativistic compositional framework of quantum operations, based on the path integral \cite{Oe:dmf,Oe:posfound}. Another advantage of this spacetime setting is that an explicit time-ordering of observables that would determine operator or operation ordering as in expression \eqref{eq:opcomp} is not required. While a proper development of the spacetime setting is out of scope for the present work, we do mention that an analog of Theorem~\ref{thm:caustransp} does hold for non-selective spacetime extended probes of the continuous outcome type, showing their causal transparency. Details will be reported elsewhere.

\subsection*{Acknowledgments}

This work was partially supported by UNAM-PAPIIT project grant IN106422 and UNAM-PASPA-DGAPA.

\subsection*{Statements and Declarations}

No datasets have been analyzed or generated in this work. The author has no competing interests to declare that are relevant to the content of this article.

\appendix


\section{Appendix}
\label{sec:app}

\begin{proof} of equation \eqref{eq:prod} of Lemma~\ref{lem:prod}. We use the representation of the sesquilinear form given by equation \eqref{eq:intrepsql} and the identity $D_{\xi}(\beta,\gamma)=\frac{\im}{2}(\{\beta,\xi\}-\{\xi,\gamma\})$ which follows from relations \eqref{eq:lip}, \eqref{eq:linobsdual} and \eqref{eq:betagammaint}.
  \begin{multline}
    (B_{\xi'}^{\epsilon'}(q')\star B_{\xi}^{\epsilon}(q))(\coh_\gamma, \coh_\beta)
    =\int_{\xL}\xd\nu(\phi)\, B_{\xi'}^{\epsilon'}(q')(\coh_\gamma,\coh_\phi) B_{\xi}^{\epsilon}(q)(\coh_\phi, \coh_\beta) \\
    =\frac{1}{\pi^2}\int_{\xL}\xd\nu(\phi)
    \int_{-\infty}^{\infty}\xd t \int_{-\infty}^{\infty}\xd t'\, e^{-(\epsilon^2 +\|\xi\|^2) t^2 -(\epsilon'^2 +\|\xi'\|^2) t'^2}
    \langle \coh_{\gamma},\coh_{\phi}\rangle\,\langle \coh_{\phi},\coh_{\beta}\rangle\\
    \exp\left(2\im t (D_{\xi}(\beta,\phi)-q)+ 2\im t' (D_{\xi'}(\phi,\gamma)-q')\right) \\
    =\frac{1}{\pi^2}
    \int_{-\infty}^{\infty}\xd t \int_{-\infty}^{\infty}\xd t'\, e^{-(\epsilon^2 +\|\xi\|^2) t^2 -(\epsilon'^2 +\|\xi'\|^2) t'^2}
    \int_{\xL}\xd\nu(\phi)
    \langle \coh_{\gamma},\coh_{\phi}\rangle\,\langle \coh_{\phi},\coh_{\beta}\rangle\\
    \exp\left(\frac12\{2t\xi,\phi\} - \frac12\{\phi,2t'\xi'\}
    +2\im t (D_{\xi}(\beta,0)-q)+ 2\im t' (D_{\xi'}(0,\gamma)-q')\right) \\
    =\frac{1}{\pi^2}
    \int_{-\infty}^{\infty}\xd t \int_{-\infty}^{\infty}\xd t'\, e^{-(\epsilon^2 +\|\xi\|^2) t^2 -(\epsilon'^2 +\|\xi'\|^2) t'^2}
    \exp\left(2\im t (D_{\xi}(\beta,0)-q)+ 2\im t' (D_{\xi'}(0,\gamma)-q')\right) \\
    \int_{\xL}\xd\nu(\phi) \langle \coh_{\gamma-2t'\xi'},\coh_{\phi}\rangle\,\langle \coh_{\phi},\coh_{\beta+2 t\xi}\rangle\\
    =\frac{1}{\pi^2}
    \int_{-\infty}^{\infty}\xd t \int_{-\infty}^{\infty}\xd t'\, e^{-(\epsilon^2 +\|\xi\|^2) t^2 -(\epsilon'^2 +\|\xi'\|^2) t'^2}
    \exp\left(2\im t (D_{\xi}(\beta,0)-q)+ 2\im t' (D_{\xi'}(0,\gamma)-q')\right) \\
    \langle \coh_{\gamma-2t'\xi'},\coh_{\beta+2 t\xi}\rangle\\
    =\frac{1}{\pi^2}
    \int_{-\infty}^{\infty}\xd t \int_{-\infty}^{\infty}\xd t'\, e^{-(\epsilon^2 +\|\xi\|^2) t^2 -(\epsilon'^2 +\|\xi'\|^2) t'^2}
    \exp\left(2\im t (D_{\xi}(\beta,0)-q)+ 2\im t' (D_{\xi'}(0,\gamma)-q')\right) \\
    \langle \coh_{\gamma},\coh_{\beta}\rangle
    \exp\left(\frac12\{2t\xi,\gamma\} -\frac12 \{\beta,2 t'\xi'\}
     -2\{\xi,\xi'\} t t'\right) \\
    =\langle \coh_{\gamma},\coh_{\beta}\rangle\frac{1}{\pi^2}
    \int_{-\infty}^{\infty}\xd t \int_{-\infty}^{\infty}\xd t'\, e^{-(\epsilon^2 +\|\xi\|^2) t^2 -(\epsilon'^2 +\|\xi'\|^2) t'^2
    -2\{\xi,\xi'\} t t'} \\
    \exp\left(2\im t (D_{\xi}(\beta,\gamma)-q)+ 2\im t' (D_{\xi'}(\beta,\gamma)-q')\right) \\
    =\langle \coh_\gamma, \coh_\beta\rangle
    \frac{1}{\pi\sqrt{(\|\xi\|^2+\epsilon^2) (\|\xi'\|^2+\epsilon'^2)- \{\xi,\xi'\}^2}}\\
    \exp\left(\frac{1}{(\|\xi\|^2+\epsilon^2) (\|\xi'\|^2+\epsilon'^2)- \{\xi,\xi'\}^2}
       \left(
       2\{\xi,\xi'\}\left(D_{\xi}(\beta,\gamma)-q\right)\left(D_{\xi'}(\beta,\gamma)-q'\right)\right. \right. \\
       \left. \left. -(\|\xi'\|^2+\epsilon'^2)\left(D_{\xi}(\beta,\gamma)-q\right)^2
       -(\|\xi\|^2+\epsilon^2)\left(D_{\xi'}(\beta,\gamma)-q'\right)^2
   \right)\right) .
  \end{multline}
\end{proof}

\begin{proof} of equation \eqref{eq:projsumid} of Lemma~\ref{lem:projsumid}.
    \begin{multline}
        \int_{-\infty}^{\infty}\xd s \left(B_{\xi'}^{\epsilon'}(q-s)\star B_{\xi}^{\epsilon}(s)\right)(\coh_\gamma,\coh_\beta)
        =\int_{-\infty}^{\infty}\xd s\,
        \langle \coh_\gamma, \coh_\beta\rangle
        \frac{1}{\pi\sqrt{(\|\xi\|^2+\epsilon^2) (\|\xi'\|^2+\epsilon'^2)- \{\xi,\xi'\}^2}}\\
        \exp\left(\frac{1}{(\|\xi\|^2+\epsilon^2) (\|\xi'\|^2+\epsilon'^2)- \{\xi,\xi'\}^2}
           \left(
           2\{\xi,\xi'\}\left(D_{\xi}(\beta,\gamma)-s\right)\left(D_{\xi'}(\beta,\gamma)-(q-s)\right)\right. \right. \\
           \left. \left. -(\|\xi'\|^2+\epsilon'^2)\left(D_{\xi}(\beta,\gamma)-s\right)^2
           -(\|\xi\|^2+\epsilon^2)\left(D_{\xi'}(\beta,\gamma)-(q-s)\right)^2
        \right)\right) \\
        =\int_{-\infty}^{\infty}\xd s\,
        \langle \coh_\gamma, \coh_\beta\rangle
        \frac{1}{\pi\sqrt{(\|\xi\|^2+\epsilon^2) (\|\xi'\|^2+\epsilon'^2)- \{\xi,\xi'\}^2}}\\
        \exp\left(\frac{1}{(\|\xi\|^2+\epsilon^2) (\|\xi'\|^2+\epsilon'^2)- \{\xi,\xi'\}^2}
          \left(-\left(2\{\xi,\xi'\}+\|\xi\|^2+\|\xi'\|^2+\epsilon^2+\epsilon'^2\right) s^2 \right.\right.\\
          +2\left( (\{\xi,\xi'\}+\|\xi'\|^2+\epsilon'^2)D_{\xi}(\beta,\gamma)
          -(\{\xi,\xi'\}+\|\xi\|^2+\epsilon^2)(D_{\xi'}(\beta,\gamma)-q)\right) s \\
          \left. \left. +2\{\xi,\xi'\}
          D_{\xi}(\beta,\gamma)\left(D_{\xi'}(\beta,\gamma)-q\right) 
          -(\|\xi'\|^2+\epsilon'^2)\left(D_{\xi}(\beta,\gamma)\right)^2-(\|\xi\|^2+\epsilon^2)\left(D_{\xi'}(\beta,\gamma)-q\right)^2
        \right)\right) \\
        =\langle \coh_\gamma, \coh_\beta\rangle
        \frac{1}{\sqrt{\pi}\sqrt{2\{\xi,\xi'\}+\|\xi\|^2+\|\xi'\|^2+\epsilon^2+\epsilon'^2}}\\
        \exp\left(\frac{1}{(\|\xi\|^2+\epsilon^2) (\|\xi'\|^2+\epsilon'^2)- \{\xi,\xi'\}^2}
         \left(\frac{1}{2\{\xi,\xi'\}+\|\xi\|^2+\|\xi'\|^2+\epsilon^2+\epsilon'^2} \right.\right.\\
         \left. \left( (\{\xi,\xi'\}+\|\xi'\|^2+\epsilon'^2)D_{\xi}(\beta,\gamma)
         -(\{\xi,\xi'\}+\|\xi\|^2+\epsilon^2)(D_{\xi'}(\beta,\gamma)-q)\right)^2\right) \\
         \left. +2\{\xi,\xi'\}
         D_{\xi}(\beta,\gamma)\left(D_{\xi'}(\beta,\gamma)-q\right) 
         -(\|\xi'\|^2+\epsilon'^2)\left(D_{\xi}(\beta,\gamma)\right)^2-(\|\xi\|^2+\epsilon^2)\left(D_{\xi'}(\beta,\gamma)-q\right)^2
         \right) \\
        =\langle \coh_\gamma, \coh_\beta\rangle
          \frac{1}{\sqrt{\pi}\sqrt{2\{\xi,\xi'\}+\|\xi\|^2+\|\xi'\|^2+\epsilon^2+\epsilon'^2}}\\
        \exp\left(\frac{1}{(\|\xi\|^2+\epsilon^2) (\|\xi'\|^2+\epsilon'^2)- \{\xi,\xi'\}^2}
        \frac{1}{2\{\xi,\xi'\}+\|\xi\|^2+\|\xi'\|^2+\epsilon^2+\epsilon'^2} \right.\\
          \left( (\{\xi,\xi'\}+\|\xi'\|^2+\epsilon'^2)^2 (D_{\xi}(\beta,\gamma))^2
          +(\{\xi,\xi'\}+\|\xi\|^2+\epsilon^2)^2(D_{\xi'}(\beta,\gamma)-q)^2 \right.\\
           \left. -2 (\{\xi,\xi'\}+\|\xi'\|^2+\epsilon'^2)(\{\xi,\xi'\}+\|\xi\|^2+\epsilon^2)D_{\xi}(\beta,\gamma)(D_{\xi'}(\beta,\gamma)-q)\right. \\
          +\left(2\{\xi,\xi'\}+\|\xi\|^2+\|\xi'\|^2+\epsilon^2+\epsilon'^2\right) \\
          \left. \left. \left(2\{\xi,\xi'\}
          D_{\xi}(\beta,\gamma)\left(D_{\xi'}(\beta,\gamma)-q\right) 
          -(\|\xi'\|^2+\epsilon'^2)\left(D_{\xi}(\beta,\gamma)\right)^2-(\|\xi\|^2+\epsilon^2)\left(D_{\xi'}(\beta,\gamma)-q\right)^2
          \right)\right)\right) \\
        =\langle \coh_\gamma, \coh_\beta\rangle
          \frac{1}{\sqrt{\pi}\sqrt{2\{\xi,\xi'\}+\|\xi\|^2+\|\xi'\|^2+\epsilon^2+\epsilon'^2}}\\
        \exp\left(-\frac{1}{2\{\xi,\xi'\}+\|\xi\|^2+\|\xi'\|^2+\epsilon^2+\epsilon'^2}
          \left(D_{\xi}(\beta,\gamma)+D_{\xi'}(\beta,\gamma)-q\right) 
          \right)
    \end{multline}
\end{proof}

\newcommand{\eprint}[1]{\href{https://arxiv.org/abs/#1}{#1}}
\bibliographystyle{stdnodoi} 
\bibliography{stdrefsb}
\end{document}